\documentclass[10pt,twocolumn,twoside]{IEEEtran}

\usepackage{cite}      
\usepackage{graphicx}  
\usepackage{psfrag}    
\usepackage{url}       
\usepackage{amssymb,amsmath,amscd,amsfonts,amstext,bbm, enumerate, ntheorem,eclbkbox}
\usepackage{array}
\usepackage{multirow}
\usepackage[latin1]{inputenc}
\usepackage{color,caption}

\usepackage{subcaption}
\usepackage{graphicx}
\usepackage[usenames, dvipsnames]{xcolor}
\usepackage{tikz}
\usepackage{pgfplots}
\usetikzlibrary{plotmarks,patterns}
	\newlength\figureheight 
	\newlength\figurewidth 
\usepackage[textwidth=1.6cm, textsize=footnotesize]{todonotes}    
\xdefinecolor{rouge}{rgb}{0.6,0.1,0}

\newcommand{\modif}[1]{{\color{black} #1}}
\newcommand{\modifPB}[1]{{\color{black} #1}}
\newcommand{\modifwh}[1]{{\color{black} #1}}

\theoremstyle{plain} 
\newtheorem{lemma}{Lemma}

\newtheorem{theorem}{Theorem}
\newtheorem{coro}{Corollary}

\theoremstyle{definition}
\newtheorem{assumption}{Assumption}

\theoremstyle{remark} 
\newtheorem{remark}{Remark}

\newcommand{\RR}{{\mathbb{R}}}
\newcommand{\NN}{{\mathbb{N}}}

\newcommand{\xlim}{{\bs x}_\star} 
\newcommand{\zetalim}{\overline{\bs \zeta}_\star} 

\newcommand{\un}{{\bs 1}}
\DeclareMathOperator*{\sprad}{{\bs r}}
\DeclareMathOperator*{\tr}{tr}
\DeclareMathOperator*{\rank}{rank}
\DeclareMathOperator*{\colsp}{span}
\newcommand{\prox}{\mathop{\mathrm{prox}}\nolimits}
\DeclareMathOperator*{\argmin}{argmin}
\DeclareMathOperator*{\circul}{circ}
\newcommand{\bs}{\boldsymbol}
\newcommand{\eqdef}{{\stackrel{\text{def}}{=}}} 
\def\adots{
  \mathinner{\mkern1mu\raise1pt\hbox{.}\mkern2mu\raise4pt\hbox{.}
  \mkern2mu\raise7pt\vbox{\kern7pt\hbox{.}}\mkern1mu}}
\def\build#1_#2^#3{\mathrel{
\mathop{\kern 0pt#1}\limits_{#2}^{#3}}}
\IEEEoverridecommandlockouts

\begin{document}

\title{Explicit Convergence Rate of a Distributed Alternating Direction 
Method of Multipliers}

\author{F. Iutzeler, P. Bianchi, Ph. Ciblat, and W. Hachem
\thanks{CNRS LTCI; Telecom ParisTech. 46 rue Barrault, 75013 Paris, 
France. 
E-mails: \texttt{forename.name@telecom-paristech.fr}. 
This work was partially granted by the French Defense Agency (DGA)
and by the Telecom/Eurecom Carnot Institute.
}} 
\date{\today} 

\markboth{Revised version}{Iutzeler, Bianchi, Ciblat and Hachem: ADMM rate of 
convergence.}

\maketitle

\begin{abstract}
Consider a set of $N$ agents seeking to solve distributively the 
minimization problem $\inf_{x} \sum_{n = 1}^N f_n(x)$ where the convex 
functions $f_n$ are local to the agents. 
The popular Alternating Direction Method of Multipliers has the
potential to handle distributed optimization problems of this kind.
\modifPB{We provide a general reformulation of the problem
and obtain a class of distributed algorithms which encompass
 various network architectures.}
The rate of convergence of \modifPB{our} method is considered. 
It is assumed that the infimum of the problem is reached at a point $x_\star$, 
the functions $f_n$ are twice differentiable
at this point and $\sum \nabla^2 f_n(x_\star) > 0$
in the positive definite ordering of symmetric matrices.
With these assumptions, it is shown that the convergence to the consensus 
$x_\star$ is linear and the exact rate is provided. 
Application examples where this rate can be optimized with respect to the ADMM free parameter $\rho$ are also given.
\end{abstract}


\section{Introduction} 
\label{intro} 

Consider a set of $N > 1$ computing 
agents that seek to solve collectively a minimization problem. 
Given that Agent $n$ has at its disposal a private convex cost function 
$f_n\,:\, \RR^K \to (-\infty, \infty]$ where $K$ is some positive 
integer, the purpose of the agents is to solve distributively the 
minimization problem 
\begin{equation}
\label{problem}
\inf_{x \in \RR^K} \sum_{n = 1}^N f_n(x) . 
\end{equation} 
A distributive (or decentralized) scheme is meant here to be an iterative 
procedure where at a each iteration, each agent updates a local estimate in 
the parameter space $\RR^K$ based on the 
sole knowledge of this agent's private cost function and on a piece of
information it 
received from its neighbors through some communication network. Eventually, 
the local estimates will converge to a common value (or consensus) which is a 
minimizer (assumed to exist) of the aggregate cost function $\sum f_n$. 

Instances of this problem appear in learning applications
where massive training data sets are distributed over a network and processed 
by distinct machines \cite{for-can-gia-jmlr10,agarwal2011reliable}, 
in resource allocation problems for communication networks 
\cite{chi-etal-procieee07,bia-jak-TAC13}, or in statistical estimation problems by
sensor networks \cite{ram-vee-ned-tac10,bia-for-hac-IT13}. 

The proximal splitting methods \cite{com-pes-tutorial11}
have recently attracted a large interest in the fields of statistics,
signal processing and communication theory thanks to their convergence
properties and to their ability to deal with large scale and
decentralized problems.  Among these, one of the most emblematic is
the Alternate Direction Method of Multipliers (ADMM). In
\cite{sch-rib-gia-sp08}, Schizas \emph{et al.}~showed that ADMM
easily lends itself to a distributed implementation of Problem
\eqref{problem}.  \modifPB{Since then, distributed versions of 
  ADMM applied to consensus problems have been explored in many works
  (see
  \cite{wei-ozd-arxiv13,iut-cdc13,shi-etal-(arxiv)13,jako-mou-xav-(arxiv)13,teixeira2013optimal,bianchi2014stochastic}
  as a non-exhaustive list).  In this paper, we provide a general
  framework inspired from \cite{sch-rib-gia-sp08} which allows to
  distribute an optimization problem on a set of agents.  From a
  formal point of view, we do \emph{not} assume the existence of a
  pre-existing ``graph'' whose edges would correspond to pairs of
  communicating agents.  Instead, our framework relies on the
  introduction of components $A_1,\dots,A_L$, each of which is a
  subset of agents. {\sl i)} In the case where the $A_\ell$'s are
  pairs of agents, our algorithm will involve pairwise communications
  between agents, as in \emph{e.g.} \cite{shi-etal-(arxiv)13}.  {\sl
    ii)} Identifying the $A_\ell$'s with larger sets of agents
  (clusters), our algorithm will be distributed at the \emph{cluster}
  level.  For instance, our framework encompasses the case of loosely
  coupled computer clusters composed of tighly coupled parallel
  machines. {\sl iii)} Finally, when the collection of components
  $A_1,\dots,A_L$ is reduced to a single set $A_1=\{1,\dots,N\}$ (that
  is, $L=1$), our algorithm reduces to the parallel ADMM algorithm
  described in \cite[Chapter 7]{boyd2011distributed}, in which all
  agents output are reduced in a centralized fashion at each iteration
  of the algorithm.  Otherwise stated, our framework yields a
  continuum of algorithms ranging from a fully centralized to a fully
  distributed setting.  }

\modifPB{The main contribution of this paper} deals with the rate of convergence of ADMM in the 
framework of Problem \eqref{problem}. 
It is assumed that the infimum of Problem \eqref{problem} is
attained at a point $x_\star$, the functions $f_n$ are twice differentiable
at this point, and $\sum \nabla^2 f_n(x_\star) > 0$ 
in the positive definite ordering of symmetric matrices. 
With these assumptions, the linear convergence of 
the ADMM iterates is shown, and most of all, their convergence rate is 
explicitly provided. Our result potentially 
allows to evaluate the impact of the communication network on the performance,
as well as the effect of the step-size. Application examples 
where the step-size can be optimized are also given. \\

The method behind the proof is as follows. We first assume that the
functions $f_n$ are quadratic. In that case, an ADMM iteration boils down to
an affine transformation that we denote as $\zeta_{k+1} = R \zeta_k + d$.  
These iterates converge at an exponential rate that can be explicitly 
obtained through an analysis of the eigenstructure of the matrix $R$. Turning 
to the general case, an ADMM iteration for $k$ large enough is shown to be 
a perturbed version of an affine transformation similar to the quadratic case.
A close look at the perturbation terms shows that they lie in such an 
eigenspace of $R$ that the analysis of the quadratic case remains essentially 
effective. 

Beyond the framework of distributed optimization, we believe that our  
technique can be used to characterize the rate of 
convergence of ADMM in more general constrained minimization settings where
the objective function is smooth in a neighborhood of the solution. \\ 

The ADMM rate of convergence was recently investigated in 
\cite{mon-svai-10,he-yuan-siam12,wei-ozd-cdc12,wei-ozd-arxiv13,gol-ma-sch-12}
where the ${\mathcal O}(1/k)$ convergence rate was established in the case
where the objective functions are not necessarily smooth. 

The authors of \cite{den-yin-rap12} consider the problem 
$
\min_{x : Ax+By = c} f(x) + g(y)
$ 
where one of the two objective functions is strongly convex and has a Lipschitz 
continuous gradient. They establish the linear convergence of the iterates
and provide upper bounds on the rate of convergence. 
The works \cite{bol-siam13} considers the quadratic or linear problem 
$\min_{x : Ax=b, x\geq 0} x^* Q x + c^* x$ where $Q$ is a symmetric positive
semidefinite matrix that may be equal to zero. The linear convergence of 
ADMM near the solution is established. \modifPB{A similar problem is investigated in~\cite{ghadimi2013}
where an upperbound on the decay rate is provided, along with the step size which minimizes the latter upperbound.}\\

The distributed consensus problem considered in this paper was also
studied by \cite{shi-etal-(arxiv)13}, \cite{jako-mou-xav-(arxiv)13} and \cite{teixeira2013optimal}. 
\modifPB{The algorithm studied by \cite{jako-mou-xav-(arxiv)13} strongly relies on 
the introduction of an inner loop at each iteration of the algorithm.
The authors of \cite{teixeira2013optimal} focus on quadratic programming and 
introduce a specific type of preconditioning for analysis purposes which also modifies the structure of the algorithm. 
Hence, both algorithms in \cite{jako-mou-xav-(arxiv)13} and 
\cite{teixeira2013optimal} differ from the natural ADMM of interest in \cite{shi-etal-(arxiv)13} 
and in the present paper.} The authors of \cite{shi-etal-(arxiv)13} prove the linear convergence of 
ADMM in a distributed setting, and provide an upper bound on the norm of the primal error.
The bound of~\cite{shi-etal-(arxiv)13} is moreover uniform w.r.t. the choice of the functions $f_n$
on a class of strongly convex functions with Lipschitz continuous gradients.
However, work is needed to fill the substantial gap between the bound of \cite{shi-etal-(arxiv)13}
and the practice. The aim of this paper is to obtain an exact and 
informative characterization of the convergence rate. In addition, the proof of \cite{shi-etal-(arxiv)13}
relies on the assumption that the functions $f_n$ are smooth strongly convex functions with 
Lipschitz continuous gradients, whereas the present paper relies on weaker assumptions.
\\

Finally, let us mention the recent preprint \cite{hon-luo-(arxiv)13} that 
considers the non smooth case. Using an approach similar to the one used in
\cite{luo-tse-93}, the linear convergence of the iterates is established in the
case where the step-size for updating the multipliers is small enough. 
No explicit convergence rate is provided. \\

After setting our assumptions in Section \ref{sec-pb}, we show how 
Problem \eqref{problem} can be distributively solved by ADMM after being
adequately reformulated. We then state our main convergence result
in Section~\ref{sec:main}. 
In Section \ref{sec-examples}, we provide an illustration of our result in some
special cases where the rate admits a simple and informative expression.
The main result is proven in Section \ref{sec-prf}. 
In Section \ref{sec-simus}, some numerical illustrations are provided. 
The conclusion is provided in Section~\ref{sec:conclusion}.

\section{Assumptions, Algorithm description}
\label{sec-pb} 

\subsection{Assumptions and problem reformulation} 
\label{formul}

\modif{Let us denote by $\Gamma_0({\mathbb R}^K)$ the set of proper, lower semicontinuous, and convex
functions from \modif{${\mathbb R}^K$} 
to $(-\infty, \infty]$ where $K$ is an integer}.
The proximity operator of a function $h \in \Gamma_0(\mathbb R^n)$ is the 
mapping defined on $\RR^n\to\RR^n$ by
\[
\text{prox}_{h}(x) = \argmin_{w} \Bigl( h(w) + 
\frac{1}{2} \| w - x \|^2 \Bigr) \,.
\]
Denote as ${\cal A} = \{1,\ldots, N\}$ the set of agents. 
The assumptions on the functions $f_n$ considered in this paper are the 
following: 

\begin{assumption}
\label{conv-lsc} 
For any $n \in {\cal A}$, $f_n \in \Gamma_0({\mathbb R}^K)$. 
\end{assumption} 

\begin{assumption}
\label{nabla2}
The infimum of the problem \eqref{problem} is attained at a point
$x_\star$. At $x_\star$, the functions $f_n$ are twice differentiable and 
their Hessian matrices satisfy 
\[
\sum_{n=1}^N \nabla^2 f_n(x_\star) > 0 . 
\]
\end{assumption} 
These assumptions clearly imply that the minimizer $x_\star$ is unique. 
Observe that the functions $f_n$ are not required to be strictly or strongly 
convex at an individual level. Moreover, no global property of the gradients 
such as the existence or the Lipschitz continuity is assumed. 
We only require the two-fold differentiability of the functions $f_n$ and the 
strong convexity of $\sum_n f_n(x)$ at a local level.

Along the idea of \cite{sch-rib-gia-sp08}, we now provide another formulation 
of Problem \eqref{problem} that
will lead us to a distributed optimization algorithm. Thanks to 
Assumption \ref{ass-equiv} below, the two formulations will be shown 
to be equivalent.  

We introduce some notations. 
Given any positive integer $\ell$, an element $x$ of $\RR^{\ell K}$ will be 
often denoted as $x = (x(1), \ldots, x(\ell) )$ where $x(m) \in \RR^K$ for 
$m=1,\ldots,\ell$. Let $C_\ell$ be the linear subspace of 
$\RR^{\ell K}$ whose elements $x = (x(1), \ldots, x(\ell))$ satisfy 
$x(1) = x(2) = \cdots = x(\ell)$. 
Denoting by $\un_\ell$ the $\ell \times 1$ vector of ones and by $\otimes$ the
Kronecker product, the orthogonal projection matrix on this subspace is 
$P_\ell = J_\ell \otimes I_K$ where $J_\ell = \ell^{-1} \un_\ell \un_\ell^*$. 
   
Given a positive integer $L$, let $A_1, \ldots, A_L$ be a collection of 
subsets of ${\cal A}$ such that the cardinality of any set $A_\ell$ 
satisfies $| A_\ell | > 1$. Define the functions 
$$
\begin{array}{clcl}  
f :& \RR^{NK} & \longrightarrow & (-\infty, \infty]  \\  
  & x & \longmapsto & 
f(x) = \sum_1^N f_n(x(n)) 
\end{array} 
$$
and 
$$
\begin{array}{clcl}  
g :& \RR^{|A_1| K} \times\cdots\times\RR^{|A_L| K} & \longrightarrow & 
                                                         (-\infty, \infty] \\  
  & z  = (z^{(1)}, \ldots, z^{(L)}) & \longmapsto & 
g(z) = \sum_1^L \imath_{C_{|A_\ell|}}(z^{(\ell)})
\end{array}
$$
where $\imath_C$ is the indicator function of $C$, defined to be equal to zero 
on $C$ and to $\infty$ outside this set.   

For any subset of agents $A \subset {\cal A}$, let 
${\cal S}_A : \RR^{NK} \to \RR^{|A|K}$ be the selection operator 
${\cal S}_A x = ( x(n))_{n\in A}$. This linear operator admits the matrix
representation ${\cal S}_A x = (S_A \otimes I_K) x$ where the matrix 
$S_A$ is a $|A| \times N$ selection matrix, \emph{i.e.}, its elements are
valued in $\{0, 1\}$, it has one non zero element per row, and
it has one non zero element at most per column. Finally, set $T = \sum_1^L | A_\ell |$ and
define the linear 
operator 
$$
\begin{array}{clcl}  
M :& \RR^{NK} & \longrightarrow & \RR^{TK} \\ 
  & x & \longmapsto & 
Mx = ({\cal S}_{A_\ell}(x) )_{\ell=1}^L = (S \otimes I_K) x 
\end{array}
$$ 
where  
\[
S = \begin{bmatrix} S_{A_1} \\ \vdots \\ S_{A_L} \end{bmatrix}
\]
is a $T \times N$ matrix. 
Operator $M$ will
be identified from now on with the matrix $M = S \otimes I_K$. 

With these definitions, we now consider the optimization problem 
\begin{equation} 
\label{pb-reseau} 
\inf_{x \in \RR^{NK}} f(x) + g(Mx) . 
\end{equation} 
Let ${\cal G} = (\{1,\ldots, L\}, {\cal E})$ be the non oriented graph with 
$\{1,\ldots, L\}$ as the set of vertices and with the set of edges $\cal E$ 
defined as $\{ \ell, m \} \in {\cal E}$ if $A_\ell \cap A_m \neq \emptyset$. 
Then, we made the following assumption. Let us remark that our proposed algorithms described later will be distributed at the subset level. A coordination will be needed within each subset $A_\ell$, but the exchanges between the subsets are fully distributed.  
\begin{assumption}
\label{ass-equiv}
The following facts hold true: 
\begin{itemize}
\item[i)] $\bigcup_{\ell=1}^L A_\ell = \cal A$, 
\item[ii)] The graph $\cal G$ is connected, 
\end{itemize} 
\end{assumption} 
We obtain the following lemma.   
\begin{lemma} 
\label{lm-equiv}
Under Assumption \ref{ass-equiv}, $x_\star$ is a minimizer of Problem
\eqref{problem} if and only if $(x_\star, \ldots, x_\star)$ is a minimizer
of Problem \eqref{pb-reseau}. 
\end{lemma}
\begin{proof}
The equivalence stated by this lemma will be established if we prove that 
$g(Mx) = \sum_{\ell=1}^L \imath_{C_{|A_\ell|}}( (x(i))_{i\in A_\ell} )$ is 
finite if and only if $x \in C_N$.  
Since $\cal G$ is connected, there exists $\ell_1 \neq 1$ such that 
$A_1 \cap A_{\ell_1} \neq \emptyset$. Therefore, 
$\imath_{C_{|A_1|}}( (x(i))_{i\in A_1} ) + 
\imath_{C_{|A_{\ell_1}|}}( (x(i))_{i\in A_{\ell_1}} ) = 
\imath_{C_{|A_1 \cup A_{\ell_1}|}}( (x(i))_{i\in A_1 \cup A_{\ell_1}} )$. 
Similarly, there exists $\ell_2 \not\in \{1, \ell_1\}$ such that 
$(A_1\cup A_{\ell_1}) \cap A_{\ell_2} \neq \emptyset$, therefore 
$\imath_{C_{|A_1|}}( (x(i))_{i\in A_1} ) + 
\imath_{C_{|A_{\ell_1}|}}( (x(i))_{i\in A_{\ell_1}} ) + 
\imath_{C_{|A_{\ell_2}|}}( (x(i))_{i\in A_{\ell_2}} ) = 
\imath_{C_{|A_1 \cup A_{\ell_1} \cup A_{\ell_2}|}}( 
(x(i))_{i\in A_1 \cup A_{\ell_1}\cup A_{\ell_2}} )$. 
Pursuing, we obtain that 
$g(Mx) = \imath_{C_{|\cup_\ell A_\ell|}}( (x(i))_{i\in \cup_\ell A_\ell} )$. 
By Assumption \ref{ass-equiv}-\emph{i)}, this is equal to 
$\imath_{C_{N}}( x )$. 
\end{proof} 

\modifPB{\subsection{An illustration} In order to be less formal and to have some insights on our formulation, consider the example given in Figure~\ref{fig:toy}.
  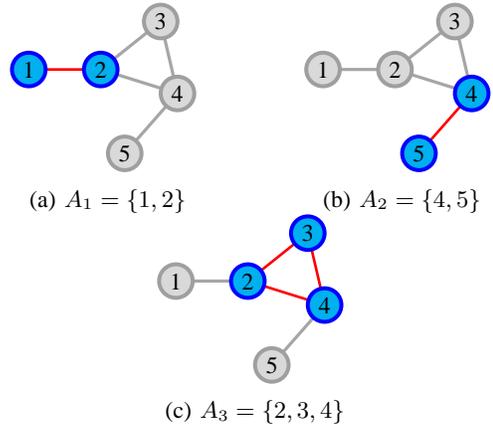
\begin{figure}[h]
	\setlength\figurewidth{0.4\columnwidth}
	\centering
        \begin{subfigure}{0.4\columnwidth}
                \centering
		\begin{tikzpicture}[x=0.9\figurewidth,y=0.9\figurewidth]


\coordinate (1) at  (0.2,0.6);
\coordinate (2) at  (0.5,0.6);
\coordinate (3) at  (0.75,0.8);
\coordinate (4) at  (0.82,0.5);
\coordinate (5) at  (0.6,0.25);

\draw[color=red,line width=1pt] (1) -- (2);
\draw[color=gray!70,line width=1.2pt] (2) -- (3);
\draw[color=gray!70,line width=1.2pt] (3) -- (4);
\draw[color=gray!70,line width=1.2pt] (2) -- (4);
\draw[color=gray!70,line width=1.2pt] (4) -- (5);

\draw [line width= 1.5pt,blue,fill=ProcessBlue] (1) circle(0.07) node[color=black,font=\small] {1};
\draw [line width= 1.5pt,blue,fill=ProcessBlue] (2) circle(0.07) node[color=black,font=\small] {2};
\draw [line width= 1.5pt,gray!75,fill=gray!30] (3) circle(0.07) node[color=black,font=\small] {3};
\draw [line width= 1.5pt,gray!75,fill=gray!30] (4) circle(0.07) node[color=black,font=\small] {4};
\draw [line width= 1.5pt,gray!75,fill=gray!30] (5) circle(0.07) node[color=black,font=\small] {5};

\end{tikzpicture}
                \caption{$A_1=\{1,2\}$}
        \end{subfigure}
~
        \begin{subfigure}{0.4\columnwidth}
                \centering
		\begin{tikzpicture}[x=0.9\figurewidth,y=0.9\figurewidth]


\coordinate (1) at  (0.2,0.6);
\coordinate (2) at  (0.5,0.6);
\coordinate (3) at  (0.75,0.8);
\coordinate (4) at  (0.82,0.5);
\coordinate (5) at  (0.6,0.25);

\draw[color=gray!70,line width=1.2pt] (1) -- (2);
\draw[color=gray!70,line width=1.2pt] (2) -- (3);
\draw[color=gray!70,line width=1.2pt] (3) -- (4);
\draw[color=gray!70,line width=1.2pt] (2) -- (4);
\draw[color=red,line width=1pt] (4) -- (5);

\draw [line width= 1.5pt,gray!75,fill=gray!30] (1) circle(0.07) node[color=black,font=\small] {1};
\draw [line width= 1.5pt,gray!75,fill=gray!30] (2) circle(0.07) node[color=black,font=\small] {2};
\draw [line width= 1.5pt,gray!75,fill=gray!30] (3) circle(0.07) node[color=black,font=\small] {3};
\draw [line width= 1.5pt,blue,fill=ProcessBlue] (4) circle(0.07) node[color=black,font=\small] {4};
\draw [line width= 1.5pt,blue,fill=ProcessBlue] (5) circle(0.07) node[color=black,font=\small] {5};

\end{tikzpicture}
                \caption{$A_2=\{4,5\}$}
        \end{subfigure}\\
        \begin{subfigure}{0.4\columnwidth}
                \centering
		\begin{tikzpicture}[x=0.9\figurewidth,y=0.9\figurewidth]


\coordinate (1) at  (0.2,0.6);
\coordinate (2) at  (0.5,0.6);
\coordinate (3) at  (0.75,0.8);
\coordinate (4) at  (0.82,0.5);
\coordinate (5) at  (0.6,0.25);

\draw[color=gray!70,line width=1.2pt] (1) -- (2);
\draw[color=red,line width=1pt] (2) -- (3);
\draw[color=red,line width=1pt] (3) -- (4);
\draw[color=red,line width=1pt] (2) -- (4);
\draw[color=gray!70,line width=1.2pt] (4) -- (5);

\draw [line width= 1.5pt,gray!75,fill=gray!30] (1) circle(0.07) node[color=black,font=\small] {1};
\draw [line width= 1.5pt,blue,fill=ProcessBlue] (2) circle(0.07) node[color=black,font=\small] {2};
\draw [line width= 1.5pt,blue,fill=ProcessBlue] (3) circle(0.07) node[color=black,font=\small] {3};
\draw [line width= 1.5pt,blue,fill=ProcessBlue] (4) circle(0.07) node[color=black,font=\small] {4};
\draw [line width= 1.5pt,gray!75,fill=gray!30] (5) circle(0.07) node[color=black,font=\small] {5};

\end{tikzpicture}
                \caption{$A_3=\{2,3,4\}$}
        \end{subfigure}
    \centering
    \caption{An example of with $L=3$ components.}
	 \label{fig:toy}
  \end{figure}
In that case, for any $x=(x(1),\dots,x(5))$, the vector $Mx$ has 3 block-components
respectively given by $(x(1),x(2))$, $(x(4),x(5))$ and $(x(2),x(3),x(4))$, that is:
\begin{equation}
Mx = \left( x(1),x(2)\, ,\, x(4),x(5)\, ,\, x(2),x(3),x(4)\right)\label{eq:3}
\end{equation}
In this example, the function $g$ is the indicator of the linear space composed of all vectors of the form
\begin{equation}
\left( u,u\, ,\, v,v\, ,\, w,w,w\right)\label{eq:2}
\end{equation}
for any $u, v,w$. This means that $g(z)$ is equal to zero whenever $z$ has the form~(\ref{eq:2})
and is equal to $+\infty$ otherwise. When $z=Mx$, we obtain that $g(Mx)$ is finite only if 
the vector~(\ref{eq:3}) has the form~(\ref{eq:2}). This holds if and only if $x(1)=x(2)$,
$x(4)=x(5)$, $x(2)=x(3)=x(4)$. Equivalently, all components of $x$ should be equal.
}

\subsection{Instancianting ADMM}
\label{subsec-decentralized}

We now recall how ADMM can be used to solve Problem 
\eqref{pb-reseau} in a distributed manner. ADMM is commonly described
by reformulating Problem \eqref{pb-reseau} into the constrained problem
\[
\inf_{z = Mx} f(x) + g(z),
\]
and by introducing the so called augmented Lagrangian. This is the function 
${\cal L}_\rho : \RR^{NK} \times \RR^{TK} \times \RR^{TK} \to 
(-\infty, \infty]$ defined as 
\[
{\cal L}_\rho(x,z,\lambda) = 
f(x) + g(z) + \langle \lambda , Mx-z \rangle + 
\frac{\rho}{2}\left\| Mx - z \right\|^2
\]
where $\rho > 0$ is a constant. ADMM consists in the iterations:
\begin{subequations}
\begin{align}
x_{k+1} &= \argmin_{x\in\RR^{NK}} {\cal L}_\rho(x, z_k; \lambda_k) \\ 
z_{k+1} &= \argmin_{z\in\RR^{TK}} {\cal L}_\rho(x_{k+1}, z; \lambda_k) \\ 
\lambda_{k+1} &= \lambda_k + \rho( M x_{k+1} - z_{k+1}). \label{lk} 
\end{align}
\end{subequations}
A proof of the following result can be found in 
\cite[Ch. 3.2 and Appendix A]{boyd2011distributed} combined with 
Lemma \ref{lm-equiv}. Another proof using the so called Douglas Rachford
splitting can be found in \cite{eckstein1992douglas}: 
\begin{theorem}
\label{1st-ord} 
Under Assumptions \ref{conv-lsc} to \ref{ass-equiv}, the set of saddle points
of the unaugmented Lagrangian ${\cal L}_0(x,z,\lambda)$ is nonempty, and any 
saddle point is of the form 
$(\un_N \times x_\star, \un_T \otimes x_\star, \lambda_\star)$ where 
$x_\star$ is the unique solution of Problem~\eqref{problem}. Moreover, 
for any initial value $(z_0, \lambda_0)$, the sequence of ADMM iterates 
$(x_k, z_k, \lambda_k)$ converges to a saddle point. 
\end{theorem} 

We now make the ADMM equations more explicit and show how they lead to a 
distributed implementation. The $x$ and $z$ -- update equations above can be 
respectively rewritten as 
\begin{align} 
x_{k+1} &= \argmin_{x\in\RR^{NK}} f(x) + 
  \frac{\rho}{2} \| Mx - (z_k - \lambda_k/\rho) \|^2 , \label{xk} \\
z_{k+1} &= \argmin_{z\in\RR^{TK}} g(z) + 
  \frac{\rho}{2} \| z - (M x_{k+1} + \lambda_k/\rho) \|^2.  \label{zk} 
\end{align} 
Let us partition $z_k$ as in the definition of the function $g$ above and write 
$z_k = (z_k^{(1)}, \ldots, z_k^{(L)})$. Accordingly, let us write 
$\lambda_k = (\lambda^{(1)}_k, \ldots, \lambda^{(L)}_k)$ where
$\lambda^{(\ell)}_k \in \RR^{|A_\ell|K}$ and furthermore, let us write 
$\lambda^{(\ell)}_k = (\lambda^{(\ell)}_k(n_1), \lambda^{(\ell)}_k(n_2), 
\ldots, \lambda^{(\ell)}_k(n_{|A_\ell|}))$
where $\lambda^{(\ell)}_k(n_i) \in \RR^K$ and where $n_i$ is the column index
of the non zero element of the row $i$ of $S_{A_\ell}$. The indices $n_i$ of
the elements of $\lambda^{(\ell)}_k$ are therefore the indices of the agents 
belonging to the set $A_\ell$. Using these notations, Equation~\eqref{zk} can 
be parallelized into $L$ equations of the form
\[
z_{k+1}^{(\ell)} = \argmin_{z\in\RR^{|A_\ell|K}} \imath_{C_{|A_\ell|}}(z) 
+ \frac{\rho}{2} \| z 
           - ({\cal S}_{A_\ell} x_{k+1} + \lambda_k^{(\ell)}/\rho) \|^2 , 
\]
whose solution is $z_{k+1}^{(\ell)} = 
\un_{|A_\ell|} \otimes \bar z_k^{(\ell)} \in C_{|A_\ell|}$ with 
\[
\bar z_{k+1}^{(\ell)} = \frac{1}{|A_\ell|} \sum_{n\in A_\ell} 
\Bigl( x_{k+1}(n) + \frac{\lambda_k^{(\ell)}(n)}{\rho} \Bigr) . 
\]
Turning to the $\lambda$ -- update equation in the ADMM iterations and 
inspecting the structure of the matrix $M$, this equation can be decomposed
into the equations 
\begin{equation} 
\label{adm-update-lambda} 
\lambda^{(\ell)}_{k+1}(n) = \lambda^{(\ell)}_{k}(n) + 
\rho ( x_{k+1}(n) - \bar z^{(\ell)}_{k+1} ) 
\end{equation} 
for $\ell = 1,\ldots, L$ and $n = 1,\ldots, N$. Fixing $\ell$ and taking the
sum of the $\lambda^{(\ell)}_{k+1}(n)$ with respect to $n$ yields 
$\sum_{n\in A_\ell} \lambda^{(\ell)}_{k+1}(n) = 0$. Therefore, the 
$\bar z_{k}^{(\ell)}$ update equation can be written after the first iteration 
as
\begin{equation}
\label{adm-update-z} 
\bar z_{k+1}^{(\ell)} = \frac{1}{|A_\ell|} \sum_{n\in A_\ell} x_{k+1}(n) . 
\end{equation}
Getting back to Equation~\eqref{xk}, we now see that it can be parallelized 
into $N$ equations of the form
\begin{equation}
\label{adm-update-x} 
x_{k+1}(n) = \argmin_{w\in \RR^K}  
f_n(w) + \!\! \sum_{m\in\sigma(n)} \!\! \langle \lambda_k^{(m)}(n), w \rangle 
+ \frac{\rho}{2} \| w - \bar z_k^{(m)} \|^2 
\end{equation} 
for $n = 1,\ldots, N$, where $\sigma(n) = \{ m \,: \,  n \in A_m \}$. 
{Let us introduce the following aggregate quantities:
\begin{eqnarray*}
  \Delta_k(n) &=& \frac{1}{\rho|\sigma(n)|}\sum_{m\in\sigma(n)}\lambda_k^{(m)}(n) \nonumber \\
  \chi_k(n) &=& \frac 1{|\sigma(n)|}\sum_{m\in\sigma(n)}\bar z^{(m)}_{k} \label{eq:adm-update-chi}
\end{eqnarray*}
After some algebra, the $x$-update in (\ref{adm-update-x}) simplifies to
\begin{equation}
  \label{eq:adm-update-x-aggrege}
  x_{k+1}(n) = \prox_{\frac{f_n}{\rho|\sigma(n)|}}\left(\chi_k(n)-\Delta_k(n)\right)\,.
\end{equation}
By~(\ref{adm-update-lambda}), we have the update equation
\begin{equation}
  \label{eq:adm-update-Delta}
  \Delta_{k+1}(n) = \Delta_k(n)  + x_{k+1}(n) - \chi_{k+1}(n) \,.
\end{equation}
We are now in position to state the main algorithm.

\subsection{Distributed ADMM (General case)}

All agents within a subset
$A_\ell$ are assumed to be connected together through a communication
network. Recall that for a given $n$, $|\sigma(n)|$ is the number of clusters
to which Agent $n$ belongs.

Before entering the iteration $k+1$, Agent $n$ holds in its memory the
values $x_k(n)$, $\chi_k(n)$ and $\Delta_k(n)$. 
\smallskip

\begin{breakbox}
\noindent {\bf Distributed-ADMM (General case)}\\
\noindent  At Iteration $k+1$, 
  \begin{enumerate}
  \item The agents $1,\ldots, N$ compute their estimates
    $x_{k+1}(n)$ 
$$
x_{k+1}(n) =
\prox_{\frac{f_n}{\rho|\sigma(n)|}}\left(\chi_k(n)-\Delta_k(n)\right)\,.
$$
\item For all $\ell=1,\ldots, L$, the agents belonging to a cluster
  $A_\ell$ send their estimates $x_{k+1}(n)$ to a ``cluster head'' who
  can be a preselected member of $A_\ell$ or an independent
  device. The cluster head computes
$$
\bar z_{k+1}^{(\ell)} = \frac{1}{|A_\ell|} \sum_{n\in A_\ell}
x_{k+1}(n) .
$$
and sends back this parameter to all the members of the cluster.
\item For $n=1,\ldots, N$, Agent $n$ computes
  \begin{eqnarray*}
    \chi_k(n) &=& \frac 1{|\sigma(n)|}\sum_{m\in\sigma(n)}\bar z^{(m)}_{k}\\
    \Delta_{k+1}(n) &=& \Delta_k(n)  + x_{k+1}(n) - \chi_{k+1}(n) \,.
  \end{eqnarray*}
\end{enumerate}
\end{breakbox}
\smallskip

Note that in the absence of a 
cluster head, one can think of a distributed computation of 
$\bar z_{k+1}^{(\ell)}$ within the cluster using \emph{e.g.} a gossiping 
algorithm. 

Note also that when $|A_\ell| = 2$ for all $\ell=1,\ldots, L$, no cluster
head nor a gossiping algorithm are needed for the execution of
Step~2. Assuming that $A_\ell = \{m, n\}$, Agents $m$ and $n$ exchange
the values of $x_{k+1}(m)$ and $x_{k+1}(n)$ then they both compute
$\bar z_{k+1}^{(\ell)} = (x_{k+1}(m) +x_{k+1}(n)) / 2$. In this case,
the algorithm is fully distributed at the agents level.  This point is
discussed in the next paragraph. }


\subsection{Distributed ADMM (Special cases)}
\label{sec:examplesCentralizedRing}

We end this section by two important examples of possible choices of the 
subsets $A_\ell$. We shall come back to these examples later. 
\paragraph*{Example 1} This is the centralized ADMM described in 
\cite[Chap. 7]{boyd2011distributed}. 
Let $L = 1$ and $A_1 = {\cal A}$. Problem \eqref{pb-reseau} becomes 
$\inf_{x\in \RR^{NK}} f(x) + \imath_{C_N}(x)$. At Iteration $k+1$, a dedicated 
device simply computes $\bar z_{k+1} = N^{-1} \sum_{n=1}^N x_{k+1}(n)$ and 
broadcasts it to all the agents. 

\paragraph*{Example 2} Here we assume that all the agents belong to a 
communication network represented by a non
oriented graph with no self loops $G = ({\cal A}, E)$ where 
$E = \{ \{ n_1, m_1 \}, \{n_2, m_2 \}, \ldots \}$ is the set of edges. 
Setting $L = | E |$, we consider that any pair of agents $\{ n, m \}$ such that 
$\{ n, m \} \in E$ is a set $A_\ell$. 
If the graph $G$ is connected, then Assumption 
\ref{ass-equiv} is easily seen to be verified, and Problems \eqref{problem} 
and \eqref{pb-reseau} are equivalent. 
{In this situation, $|\sigma(n)|$ simply coincides with the \emph{degree $d_n$ of
node $n$} in the graph \emph{i.e.}, the number of its neigbors.
For every edge $A_\ell=\{n,m\}$, $\bar z^{(\ell)}_k = (x_k(n)+x_k(m))/2$ is simply the average of the nodes' estimates on that edge.
As a consequence, it is straightforward to show that
$\chi_k(n) = (x_k(n)+\bar x_k(n))/2$ where 
$\bar x_k(n)=\frac 1{d_n}\sum_{m\in {\cal N}_n}x_k(m)$ and ${\cal N}_n$ is the neighborhood of node $n$.

This leads to the following algorithm. An Agent $n$ keeps the variables $x_k(n),\bar x_k(n), \Delta_k(n)$
at each time $k$.
\smallskip

\begin{breakbox}
\noindent {\bf Distributed-ADMM (clusters are edges)}\\
\noindent  At Iteration $k+1$, Agent $n$ 
  \begin{enumerate}
  \item computes its estimate
$$
x_{k+1}(n) =
\prox_{\frac{f_n}{\rho d_n}}\left(\frac{x_k(n)+\bar x_k(n)}2-\Delta_k(n)\right)\,,
$$
\item receives the estimates $x_k(m)$ of other Agents $m\in {\cal N}_n$ in its neighborhood and computes
$$
\bar x_k(n)=\frac 1{d_n}\sum_{m\in {\cal N}_n}x_k(m)\,,
$$
\item updates $\Delta_{k+1}(n) = \Delta_k(n) + (x_{k+1}(n)-\bar x_{k+1}(n))/2$.
\end{enumerate}
\end{breakbox}
\smallskip

In this special case, the algorithm boils down to the algorithm of \cite{shi-etal-(arxiv)13}.}

\section{Main result} 
\label{sec:main}

We now come to the main result of this paper. Define the $T\times T$ 
orthogonal projection matrix 
\[
\Pi = \begin{bmatrix} J_{|A_1|} & \\ & \ddots \\ & & J_{|A_L|} \end{bmatrix} 
\]
and consider the $TK \times TK$ orthogonal projection matrix 
\[
P = \Pi \otimes I_K = \begin{bmatrix} P_{|A_1|} 
\\ & \ddots \\ & & P_{|A_L|} \end{bmatrix}  . 
\]
Define the $TK \times TK$ matrix 
\begin{align} 
Q &= 
\rho M \left( 
\begin{bmatrix} \nabla^2 f_1(x_\star) \\ & \ddots \\ & & 
\nabla^2 f_N(x_\star) \end{bmatrix} + \rho M^* M \right)^{\!\!\!-1} M^* 
\nonumber \\
&=  
\rho M \left( \nabla^2 f( \un_N \otimes x_\star) + \rho M^* M \right)^{-1} M^* 
\label{def-Q} 
\end{align} 
Finally, denote by $\colsp(\cdot)$ and by $\sprad(\cdot)$ respectively the 
column space and the spectral radius of a matrix. 

\begin{theorem}
\label{main} 
Let Assumptions \ref{conv-lsc} to \ref{ass-equiv} hold true. 
Let $\alpha = \sprad( (\Pi_{\colsp(P+Q)} - (P+Q) )(I-2P) )$  
where $\Pi_{\colsp(P+Q)}$ is the orthogonal projection matrix on $\colsp(P+Q)$.
Then the following facts hold true: 
\begin{itemize}
\item[i)] 
$\alpha < 1$, 
\item[ii)] 
For any initial value $(z_0, \lambda_0)$ of ADMM, 
\[
\limsup_{k\to\infty} \frac 1k \log \| x_k - \un_N \otimes x_\star \| 
\leq \log\alpha , 
\]
\item[iii)] 
The matrix $R = (I - P - Q)(I-2P)$ has eigenvalues with absolute value 
$\alpha$. If $\colsp((I-2P) M)$ is not orthogonal to the invariant subspace
of $R$ associated with these eigenvalues, then 
\[
\limsup_{k\to\infty} \frac 1k \log \| x_k - \un_N \otimes x_\star \| 
= \log\alpha  
\]
for any initial value $(z_0, \lambda_0)$ outside a set of Lebesgue measure equal to zero.
\end{itemize} 
\end{theorem} 
This theorem says that $\| x_k - \un_N \otimes x_\star \| < 
(\alpha + o(1))^k$, and provides a condition under which this rate is tight. 
It means that, for $(z_0, \lambda_0)$ outside a set with zero Lebesgue 
measure, $\| x_k - \un_N \otimes x_\star \|=(\alpha + o(1))^k$.
Moreover, Item iii) states that this rate is \emph{tight} as soon as matrices $P$ and $M$ satisfies a technical condition.
Although providing deeper insights on this condition is a difficult task, we claim that the latter condition is mild. 
It is for instance trivially satisfied in the case of a centralized network (see Example~1 in Section~\ref{sec:examplesCentralizedRing}).

\section{Special cases}
\label{sec-examples} 

The aim of this section is to provide some examples of the rate $\alpha=
\sprad( (\Pi_{\colsp(P+Q)} - (P+Q) )(I-2P) )$ in simple scenarios.  We assume
for simplicity that the dimension $K$ of the parameter space is equal to one,
in which case we have $M = S$ and $P = \Pi$.  Moreover, we investigate the case
where $f''_n(x_\star)=\sigma_\star^2>0$ is a constant which does not depend on
$n$.  Remark that this assumption is not mandatory, but has the benefit of
yielding simple and insightful expressions.

\subsection{The centralized network}
\label{subsec-centralized} 

We consider here the simple configuration of Example~$1$ of 
Section~\ref{subsec-decentralized}, where $L = 1$ and $A_1 = {\cal A}$. 
This case amounts to assuming that $M = I_N$, $P = N^{-1} \un_N \un_N^*$
and $Q=\frac{\rho}{\sigma_\star^2+\rho}I_N$. The projector $\Pi_{\text{span}(P+Q)}$
is the identity and the rate of convergence $\alpha$ of 
ADMM coincides with the spectral radius of 
\begin{align*} 
R &= \Bigl( I_N - P - Q \Bigr)\Bigl( I_N - 2P \Bigr) \\
 &=   \frac{\sigma_\star^2}{\sigma_\star^2+\rho}(I_N-P) +\frac{\rho}{\sigma_\star^2+\rho}P \,.
\end{align*} 
Matrix $R$ has two possibly distinct eigenvalues $\frac{\sigma_\star^2}{\sigma_\star^2+\rho}$ and 
$\frac{\rho}{\sigma_\star^2+\rho}$ both of them less than one. We have the following corollary.

\begin{coro}[Centralized network]
Under the stated assumptions, the rate is given by: 
  \begin{equation}
    \label{eq:1}
    \alpha = \frac{\max(\rho,\sigma_\star^2)}{\rho+\sigma_\star^2}\,.
  \end{equation}
In particular, $\alpha\geq \frac 12$ with equality iff $\rho=\sigma_\star^2$.
\label{coro:centralized}
\end{coro}
Corollary~\ref{coro:centralized} states that the optimal convergence rate is $1/2$ and that 
this rate is attained when the step-size of the algorithm is equal to $\sigma_\star^2$.
In general, $\sigma_\star^2$ is unknown, but the result nevertheless provides useful guidelines on the way to select parameter $\rho$,
and potentially allows to design adaptive step-size selection strategies.

It is worth noting that a closed-form expression of the rate $\alpha$ can as well be obtained in the case where
the second order derivatives $f''(x_\star)$ are distinct. The analysis is however somewhat tedious, as the latter closed-form expression
depends on the location of $\rho$ on the real axis. We shall discuss this case below in the numerical section.

%

\subsection{The ring network}
\label{subsec-ring} 

We now assume that $N \geq 3$ and consider the framework of Example $2$ of 
Section~\ref{subsec-decentralized}. In that framework, the graph 
$G = ({\cal A}, E)$ that we study here is the ring network with 
$E = \{ \{ 1,2 \},$ $\{2,3\},\ldots, \{N-1, N\}, \{ N, 1 \} \}$
as the set of edges. We therefore have $L=N$ and 
\[
A_\ell = \left\{ \begin{array}{ll} \{ \ell, \ell + 1 \} 
& \text{if} \ \ell < N \\ 
\{ 1, N \} & \text{otherwise}   
\end{array} \right.  
\]
as shown in Figure \ref{fig:ring}. 
\begin{figure}[h]
  \centering
  \includegraphics[width=0.4\linewidth]{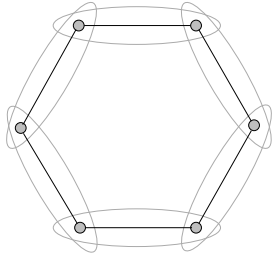}
  \caption{Example of a ring network with $N=6$. Sets $A_\ell$ 
are represented by the ellipses.} 
  \label{fig:ring}
\end{figure}

We define for simplicity $s_N=\sin(2\pi/N)$, $c_N=\cos(2\pi/N)$ and $t_N=\tan(2\pi/N)$.
\begin{coro}[Ring network]
  Under the stated assumptions, the rate $\alpha=\alpha(\rho)$ is given by the following expression.
  \begin{itemize}
  \item If $\displaystyle{\rho \leq \frac{\sigma_\star^2}{2s_N}}$, then
    \[
    \alpha = \frac{\sigma_\star^2 + 2\rho\left( 1 +c_N\right) + \sqrt{\sigma_\star^4 - 4 \rho^2s_N^2}} {2(\sigma_\star^2 + 2\rho)} ,
    \]
  \item If $\displaystyle{\rho \in \Bigl[\frac{\sigma_\star^2}{2 s_N}, \frac{\sigma_\star^2}{2 t_N^2} \Bigr]
    }$, then
    \[
    \alpha = \sqrt{\frac{\rho(1+c_N)}{\sigma_\star^2 + 2\rho} } ,
    \]
  \item If $\rho\geq \frac{\sigma_\star^2}{2 t_N^2}$ then
    \[
    \alpha = \frac{2\rho}{\sigma_{\star}^2 + 2\rho} .
    \]
  \end{itemize}
For any $N\geq 3$, the function $\rho\mapsto \alpha(\rho)$ is continuous, decreasing on $(0,\frac{\sigma_\star^2}{2s_N}]$,
increasing on $[\frac{\sigma_\star^2}{2s_N},+\infty)$. Finally, 
$$
\alpha\geq \alpha_{\text{opt}} := \frac{1}{\sqrt 2}\sqrt{\frac{1+c_N}{1+s_N}}
$$ 
with equality iff $\rho = \frac{\sigma_\star^2}{2s_N}$.
\label{coro:ring}
\end{coro}

\begin{proof}
  See Appendix.
\end{proof}

The optimal step-size $\rho_{\text{opt}} = \frac{\sigma_\star^2}{2s_N}$ is 
equal to $\frac{\sigma_\star^2N}{4\pi}+o(N)$ which suggests that the step-size 
should increase at the rate $N$. In that case, we get
\begin{equation}
\label{as-ring} 
\alpha_{\text{opt}}  = 1-\frac\pi N+o(\frac 1{N}). 
\end{equation}

Before closing this section, it is interesting to compare these results with
the speeds of some well-known algorithms used for solving the so-called
\emph{average consensus} problem (see the seminal work of \cite{boyd-it06} for
more details).  Given a network of $N$ agents each holding a measurement
$\theta_n$, the purpose of these agents is to reach a consensus over the
average $N^{-1} \sum_{n=1}^N \theta_n$. 
To attain this consensus, one idea is
to solve the problem $\min_x \sum_{n=1}^N (x-\theta_n)^2$ by means of ADMM.
This approach has been undertaken in~\cite{ers-etal-sp11}, where the authors
also obtained Equation~\eqref{as-ring} in the large $N$ asymptotic regime of
the ring network. 
{When synchronous \emph{gossip} algorithms for average
consensus is considered, it is proven in~\cite{car-automatica08} that $\alpha
\simeq 1 - 2 \pi^2 / N^2$ for large $N$ ring graphs.}  

\section{Proof of Theorem \ref{main}} 
\label{sec-prf} 

{
The proof is centered around the vector sequence 
$\zeta_k = \lambda_k + \rho z_k$ who can be shown to govern the evolution of
ADMM. Starting with the case where all the functions $f_n$ are quadratic, 
we show that an ADMM iteration boils down to the affine transformation 
$\zeta_{k+1} = R \zeta_k + d$ where $R$ is a $TK \times TK$ matrix with a 
certain structure. 
A spectral analysis of $R$ shows then that the 
rate of convergence of ADMM is provided by the largest modulus of the 
eigenvalues of $R$ different from one (this number is smaller that one). 
Part of the proof consists in characterizing the eigenspace associated with 
these eigenvalues. We then generalize our results to the case where the 
functions $f_n$ are not necessarily quadratic, but remain twice differentiable
in the vicinity of the minimum. In this case, we obtain a perturbed version
of an affine transformation similar to the quadratic case. The perturbation 
terms will be shown to lie in such a subspace of $R$ that the analysis of 
the quadratic case remains essentially effective. 
}

We start our proof by providing preliminary results that describe some simple
algebraic properties of the matrices $S$ and $\Pi$ {who will help us 
study the eigenstructure of $R$}. We then recall some well 
known properties of the \emph{proximity operator} which is known to be tightly 
related with ADMM. 
We then establish Theorem~\ref{main} in the quadratic case, ending our
proof with the general case.

\subsection{Preliminary results}

We shall need to reformulate Assumption~\ref{ass-equiv} in a form that will be 
more conveniently used in the proof: 
\begin{lemma}
\label{lm-span}
Assumption \ref{ass-equiv} is equivalent to: 
\begin{itemize} 
\item[i)] $\rank(S) = N$, and 
\item[ii)] $\displaystyle{ \colsp(S) \cap \colsp(\Pi) = \colsp(\un_T) }$.  
\end{itemize} 
\end{lemma}
\begin{proof}
Since the matrix $S$ has one non zero element per row, its non zero columns 
are linearly independent. Assumption \ref{ass-equiv}-\emph{i)} is equivalent 
to the fact that no column of 
$S$ is zero. Therefore, Assumptions \ref{ass-equiv}-\emph{i)} and 
Item $i)$ in the statement of the lemma are equivalent.  

Denoting any vector $v \in \RR^N$ as $v = (v_1,\ldots, v_N)$, we have 
\begin{multline*} 
y \in \colsp(S) \cap \colsp(\Pi) 
\Leftrightarrow  \\  
\exists v \in \RR^N \, : \,   
y = ( (v_n)_{n\in A_1}, \ldots, (v_n)_{n\in A_L} )  \\ 
 \text{and} \ 
\sum_{\ell=1}^L \imath_{\colsp(\un_{|A_\ell|})}((v_n)_{n\in A_\ell}) = 0 . 
\end{multline*} 
If Assumption \ref{ass-equiv}-\emph{ii)} is satisfied, then the proof
of Lemma \ref{lm-equiv} shows that 
$\sum_{\ell} \imath_{\colsp(\un_{|A_\ell|})}((v_n)_{n\in A_\ell}) = 
\imath_{\colsp(\un_{|\cup_\ell A_\ell|})}( (v_n)_{n\in \cup_\ell A_\ell})$, 
which shows that $y \in \colsp(\un_T)$. Conversely, suppose that 
Assumption \ref{ass-equiv}-\emph{ii)} is not satisfied. 
Then there exists a non empty set $\cal C$ strictly included in 
$\{1,\ldots, L \}$ such that $\cup_{\ell \in {\cal C}} A_\ell$ and 
$\cup_{\ell \in \{1,\ldots, L\} - {\cal C}} A_\ell$ are disjoint. 
Let $v \in \RR^N$ be defined as 
$( v_n )_{n \in \cup_{\ell \in {\cal C}} A_\ell} = 
\alpha \un_{| \cup_{\ell \in {\cal C}} A_\ell|}$ and 
$( v_n )_{n \in \cup_{\ell \in \{1,\ldots, L\} - {\cal C}} A_\ell} = 
\beta \un_{| \cup_{\ell \in \{1,\ldots, L\} - {\cal C}} A_\ell|}$ 
with $\alpha \neq \beta$. Then it is easy to see that 
$y = ( (v_n)_{n\in A_1}, \ldots, (v_n)_{n\in A_L} ) \in 
\colsp(S) \cap \colsp(\Pi)$ but $y \not\in \colsp(\un_T)$. In conclusion, 
Assumption \ref{ass-equiv}-\emph{ii)} and Item $ii)$ are
equivalent. 
\end{proof} 

Other properties of $S$ will be needed: 
\begin{lemma}
\label{lm-S} 
The matrix $S$ satisfies the following properties:  
\emph{(i)} $S \un_N = \un_T$, \emph{(ii)} $(I_T - \Pi) S \un_N = 0$.
\end{lemma} 
\begin{proof}
The property \emph{(i)} is due to the fact that $S$ contains one non zero 
element per row, and this element is equal to $1$. Since 
$\Pi \un_T = \un_T$, we obtain \emph{(ii)}. 
\end{proof} 

\subsection{The proximity operator}

{Let $h \in \Gamma_0(\mathbb R^n)$.}
The following two lemmas are well known, see 
\emph{e.g.}~\cite{livre-combettes}: 
\begin{lemma}
\label{lipsh} 
The operator $\prox_h$ is non expansive, \emph{i.e.}, 
$\| \prox_h(x) - \prox_h(x') \| \leq \| x - x' \|$.
\end{lemma}

\begin{lemma}
\label{proj} 
Let ${\cal C} \subset \RR^n$ be a closed convex set. Then 
$\prox_{\imath_{\cal C}}(x)$ coincides with the orthogonal projection of $x$ 
on $\cal C$.
\end{lemma}

\subsection{The quadratic case} 

We consider herein the case where the functions $f_n$ are quadratic: 
$f_n(x(n)) = 0.5 \, x(n)^* \Phi_n  x(n) + c_n^* x(n) + d_n$ where the 
$K\times K$ matrices $\Phi_n$ are symmetric and nonnegative, the $c_n$ are 
some given $K\times 1$ vectors, and the $d_n$ are some given scalars. 
Assumption 
\ref{nabla2} reads in this case $\sum_1^N \Phi_n > 0$. We immediately
observe that the solution of Problem \eqref{problem} is attained at a unique
point $x_\star$. Observe also that the function $f(x)$ is quadratic with
the gradient $\nabla f(x) = \Phi x + c$ where
\begin{equation}
\label{def-Sigma} 
\Phi = \begin{bmatrix} \Phi_1 \\ & \ddots \\ & & \Phi_N \end{bmatrix}, 
\quad \text{and} \quad  
c = \begin{bmatrix} c_1 \\ \vdots \\ c_N \end{bmatrix} . 
\end{equation}  
Our first task is formulate an ADMM iteration as a single line affine 
transformation, as alluded to in the introduction of this paper. 
We start with the $z$ -- update equation~\eqref{zk}. Writing 
$\zeta_{k+1} = \rho M x_{k+1} + \lambda_{k}$, Equation~\eqref{zk} can be 
rewritten $z_{k+1} = \prox_{\rho^{-1} g}(\zeta_{k+1} / \rho)$.
Recall now that the matrices $P_{|A_\ell|}$ introduced in 
Section~\ref{formul} are the orthogonal projection matrices 
on the subspaces $C_{|A_\ell|}\subset \RR^{|A_\ell|K}$. 
Therefore, the matrix $P$ is the orthogonal projection matrix on the 
subspace of $\RR^{TK}$ that coincides with the domain of $g$. By 
Lemma~\ref{proj}, we get that $z_{k+1} = \rho^{-1} P \zeta_{k+1}$. 

Letting $P_\perp = I_{TK} - P$ be the orthogonal projection matrix on the 
orthogonal complement of $\colsp(P)$, we get that 
$P_\perp\zeta_{k+1} = \zeta_{k+1} - P \zeta_{k+1} = 
\zeta_{k+1} - \rho z_{k+1} = \lambda_k + \rho( Mx_{k+1} - z_{k+1}) = 
\lambda_{k+1}$. Summarizing, we have 
$\zeta_k = \lambda_k + \rho z_k$, $\rho z_k = P \zeta_k$ and 
$\lambda_k = P_\perp \zeta_k$ for any $k\in\NN$. 

We now turn to the $x$ -- update equation. Since $f$ is differentiable, 
Equation~\eqref{xk} can be rewritten as 
\[
\nabla f(x_{k+1}) + \rho M^* (M x_{k+1} + \lambda_k/\rho - z_k ) = 0 , 
\]
or equivalently, 
\[
\nabla f(x_{k+1}) + \rho M^* M x_{k+1} + M^* (I-2P) \zeta_k = 0 . 
\]
As $\nabla f(x) = \Phi x + c$, we get 
$(\Phi + \rho M^* M ) x_{k+1} = - M^* (I-2P) \zeta_k - c$. 
Since $M^* M = (S^* S) \otimes I_K > 0$ by Lemma \ref{lm-span}, the matrix 
$H = \Phi + \rho M^* M$ is invertible, and we end up with 
\begin{equation}
\label{x-quad} 
x_{k+1} = - H^{-1} M^* (I - 2P) \zeta_k - H^{-1} c . 
\end{equation} 
Recalling that $\zeta_{k+1} = \rho M x_{k+1} + \lambda_k$ and observing that
the matrix $Q$ defined in~\eqref{def-Q} coincides with $\rho M H^{-1} M^*$ in 
the quadratic case, we finally obtain 
\begin{align}
\nonumber \zeta_{k+1} &= - Q(I-2P) \zeta_k + \lambda_k - \rho M H^{-1} c \\
\nonumber             &= (P_\perp - Q)(I-2P) \zeta_k - \rho M H^{-1} c \\ 
            &\eqdef R \zeta_k + d  \label{eq:rec-R}
\end{align} 
where $R=(P_\perp - Q)(I-2P)$ and $d = - \rho M H^{-1} c$. 

\begin{remark}
\label{DR} 
These derivations show that the sequence $\zeta_k$ is autonomous and 
completely characterizes ADMM. This phenomenon is in fact general and
shows up naturally when ADMM is interpreted as a particular case of 
the Douglas-Rachford splitting algorithm 
\cite{gabay1976dual,eckstein1992douglas}. 
\end{remark}

The following lemma provides some important spectral properties of $R$. In 
fact, part of the results shown in its statement can be deduced from 
Theorem~\ref{1st-ord}. Indeed, a consequence of this theorem is that 
the iterations $\zeta_{k+1} = R \zeta_k + d$ converge for any initial value 
$\zeta_0$. Yet, the direct proof provided below provides a finer understanding 
of the spectral properties of $R$: 

\begin{lemma}
\label{lm-R} 
$\sprad(R)\leq 1$. Moreover, for any $\theta \in (0,2\pi)$, 
$\exp(\imath\theta)$ is not an eigenvalue of $R$. Finally, the algebraic
and geometric multiplicities of any eigenvalue of $R$ coincide. 
\end{lemma} 

\begin{proof}
We have $\sprad(R)\leq \| R \| = \| P_\perp - Q \|$.
\modif{Upon noting that $|x^* P_\perp x - x^* Q x| \leq \| x \|^2$ for any $x \in \RR^{TK}$,
all eigenvalues of the real symmetric matrix $P_\perp - Q$ have their absolute value no larger than one.
Thus, the same holds for $(P_\perp - Q)^2$ from which we obtain $\| P_\perp - Q \|\leq 1$. This proves the first point of the Lemma.}

Assume that $\exp(\imath\theta)$ is an eigenvalue of $R$ for some
$\theta \in [0,2\pi)$, and let $w$ be an associated eigenvector with
$\| w \| = 1$. Since $Rw = \exp(\imath\theta) w$, we have $\| Rw \| =
1$, which implies that $w_{\text{R}}^* (P_\perp - Q)^2 w_{\text{R}} =
1 = \| w_{\text{R}} \|^2$ where $w_{\text{R}} = (I - 2P)w$.
\modif{This implies that $w_{\text{R}}$ lies in the eigenspace of
  $(P_\perp-Q)^2$ corresponding to the unit eigenvalue.  Any such
  vector can be written as the sum $w_{\text{R}}=u+v$ of two
  orthogonal vectors satisfying $(P_\perp-Q)u=u$ and
  $(P_\perp-Q)v=-v$\,. As $u^*(P_\perp-Q)u=\|u\|^2$, one has
  $\|u\|^2\geq u^*P_\perp u = \|u\|^2 + u^*Qu$. Since $Q$ is non
  negative, we obtain $Qu=0$ and, consequently, $Pu=0$.  Using similar
  arguments,  $P_\perp v=0$. Now, equation $Rw = \exp(\imath\theta) w$ reads $(P_\perp
  - Q)(u+v) = \exp(\imath\theta) (I-2P)(u+v)$. In the light of the above properties of $u$ and $v$,
this is equivalent to
  $u-v=\exp(\imath\theta)(u-v)$. Thus $\theta=0$ and the second point is proved.}  

The eigenvalue $\lambda$ has the same algebraic and geometric
multiplicities if and only if $\rank((R - \lambda I)^2) = \rank(R -
\lambda I)$ (indeed, any Jordan block ${\cal J}_\lambda$ with size $>
1$ associated with $\lambda$ would satisfy $\rank(({\cal J}_\lambda -
\lambda I)^2)= \rank({\cal J}_\lambda -\lambda I) - 1$).  Using the
identities $\rank(AB) = \rank(BA)$ and $\rank(AB) = \rank(B)$ if $A$
is invertible, we obtain
\[
\rank(R - \lambda I) = \rank( P_\perp - Q  - \lambda (P_\perp - P) ) 
\]
since $(P_\perp - P)^2 = I$, and  
\begin{align*} 
\rank((R - \lambda I)^2) &= 
\rank( [( P_\perp - Q  - \lambda (P_\perp - P)) (P_\perp - P)]^2 )  \\
&= 
\rank( [ P_\perp - Q  - \lambda (P_\perp - P) ]^2 )  . 
\end{align*} 
Since $P_\perp - Q  - \lambda (P_\perp - P)$ is symmetric, these two ranks 
coincide. 
\end{proof} 

When it exists, the eigenspace of $R$ associated with the eigenvalue $1$ 
plays an important role. In order to characterize this subspace, we start
with a preliminary result: 
\begin{lemma}
\label{lm-N}
Let Assumptions \ref{nabla2} and \ref{ass-equiv} hold true, and let 
${\cal N} = \ker( Q - P)$ be the null space of $Q - P$. Then 
\[
{\cal N} = \ker(Q) \cap \ker(P) = \ker(Q+P) .   
\]
\end{lemma}
\begin{proof}
\modif{We provide the proof of the first equality ${\cal N} = \ker(Q) \cap \ker(P)$
(the second equality follows from the non negativity of $Q$ and $P$). Note 
that we only need to prove that
${\cal N} \subset \ker(Q) \cap \ker(P)$, the other inclusion being trivial.
For any $\zeta\in {\cal N}$, we set $\zeta = \lambda + \rho z$ where  $\lambda = P_\perp \zeta$
and $\rho z= P\zeta$. Vector $\zeta$ satisfies $Q\zeta = P \zeta = \rho z$.
}
 Observe that $\colsp(Q) = \colsp(M) = \colsp( S \otimes I_K) = 
\colsp(S) \otimes \RR^K$ 
where the second $\otimes$ denotes the tensor product. 
And since $\colsp(P) = \colsp( \Pi \otimes I_K) = \colsp(\Pi) \otimes \RR^K$, 
we obtain that  
$z \in \colsp(M) \cap \colsp(P) = \left(\colsp(S) \cap \colsp(\Pi)\right)
\otimes \RR^K = \colsp(\un_T) \otimes \RR^K$ by Lemma \ref{lm-span}. 
Hence $z = \un_T \otimes q$ where $q \in \RR^K$ needs to be determined. 
Replacing in the equation $Q\zeta = \rho z$, we obtain 
$M H^{-1} \left( M^* \lambda + \rho M^* (\un_T \otimes q) \right) = 
\un_T \otimes q$. But 
$M (\un_N \otimes q) = (S \otimes I_K) (\un_N \otimes q) 
= S \un_N \otimes q = \un_T \otimes q$ by Lemma \ref{lm-S}. Since $M$ is 
full column rank by Lemma \ref{lm-span}, we obtain 
$M^* \lambda + \rho M^* (\un_T \otimes q) = H (\un_N \otimes q)$, or 
\begin{align*}
M^* \lambda &= \Phi (\un_N \otimes q) + 
\rho M^* (S \otimes I_K) (\un_N \otimes q) - \rho M^* (\un_T \otimes q) \\ 
&= \Phi (\un_N \otimes q) 
\end{align*}
by Lemma \ref{lm-S}. This lemma also shows that 
$(\un_N^* \otimes I_K) M^* \lambda = 
(\un_N^* S^* \otimes I_K) ((I_T - \Pi) \otimes I_K) \lambda = 0$. Hence 
$(\un_N^* \otimes I_K) \Phi (\un_N \otimes q) = \sum_1^N \Phi_n q = 0$, 
which implies $q = 0$ by Assumption \ref{nabla2}. 
This shows that $Q \zeta = P \zeta = 0$, in other words 
$\zeta \in \ker(Q) \cap \ker(P)$. 
\end{proof}

Recall from Lemma \ref{lm-R} that all the Jordan blocks of $R$ are trivial. 
Denoting as $\dim(\cdot)$ the dimension of a vector space, we have: 
\begin{lemma}
\label{vap-1}
The matrix $R$ has an eigenvalue equal to $1$ if and only if 
$\dim({\cal N}) > 0$. In that case, let 
\[
R = W \Lambda W^{-1} = 
\Bigl[ W_1 \ W_2 \Bigr] 
\begin{bmatrix} I & \\ & \tilde\Lambda \end{bmatrix} 
\begin{bmatrix} \ \underline W_1^* \ \\ \ \underline W_2^* \ \end{bmatrix} 
\]
be a spectral factorization of $R$. 
Then $W_1 \underline W_1^* = 
\Pi_{\cal N}$, the orthogonal projection matrix onto $\cal N$. 
Whether $R$ has or has not an eigenvalue equal to $1$ (in which case 
we set $\Pi_{\cal N} = 0$), 
\[
R - \Pi_{\cal N} =  (\Pi_{\colsp(P+Q)} - (P+Q) )(I-2P)  . 
\]
\end{lemma} 
\begin{proof}
When $R$ has an eigenvalue equal to $1$, any vector $w$ of the associated right 
eigenspace ${\cal E}_{\text{right}}(1)$ satisfies 
$(P_\perp - Q) (I - 2P) w = w$. Writing $w_{\text{R}} = (I - 2P) w$ and 
recalling that $(I-2P)^2 = I$, we obtain 
$w \in {\cal E}_{\text{right}}(1) \Leftrightarrow (Q - P) w_{\text{R}}=0$, 
in other words $w_{\text{R}} \in \cal N$. But since 
${\cal N} \subset \ker(P)$ by Lemma \ref{lm-N}, we obtain that
${\cal E}_{\text{right}}(1)={\cal N}$. 
We show similarly that the left eigenspace associated with the eigenvalue 
$1$ is $\cal N$. 

Turning to the spectral factorization, since 
$\colsp(W_1) = \colsp(\underline W_1) = \cal N$, we can write 
$\underline W_1 = W_1 U$ where $U$ is an invertible matrix. Since  
$\underline W_1^* W_1 = I_{\dim(\cal N)}$, we have 
$U = (W_1^* W_1)^{-1}$ which shows that $W_1 \underline W_1^* = \Pi_{\cal N}$. 

Finally, since $(I-2P)^2 = I$ and ${\cal N} \subset \ker(P)$, 
\begin{align*} 
R - \Pi_{\cal N} &= ( I - (P+Q) - \Pi_{\cal N} (I - 2P) ) (I-2P) \\ 
&= ( I - (P + Q + \Pi_{\cal N}) )(I - 2P) \\
&= (\Pi_{\colsp(P+Q)} - (P+Q) )(I-2P)  . 
\end{align*} 
\end{proof} 

The spectral properties of $R$ that we just established lead us to the 
following lemma. We denote by $A^\sharp$ the Moore-Penrose pseudo-inverse of 
the matrix $A$. 
\begin{lemma}
\label{cvg-quad} 
Define $\zetalim = \rho (I - 2P) (P-Q)^\sharp M H^{-1} c$. The set of 
fixed points of  the transformation $\zeta\mapsto R\zeta + d$ coincides with 
$\{\zetalim\}+{\cal N}$.
\end{lemma}
\begin{proof}
Let $\bar\zeta$ be a fixed point of the transformation $\zeta' = R\zeta + d$.
Using the identity $(I-2P)^2 = I$, the equation $(I-R) \bar\zeta = d$ 
reads $(Q - P) \bar\zeta_{\text{R}} = -\rho M H^{-1} c$ where
$\bar\zeta_{\text{R}} = (I - 2P) \bar\zeta$. 
\modif{Note that $ M H^{-1} c\in \colsp{Q}$. By Lemma~\ref{lm-N}, ${\cal N}=\ker{(Q-P)}\subset \ker Q$, thus 
$u^TMH^{-1}c=0$ for any $u\in {\cal N}$. This means that $ M H^{-1} c\in \colsp{(Q-P)}$.
Consequently, the set of solutions to $(Q - P) \bar\zeta_{\text{R}} = -\rho M H^{-1} c$ is nonempty 
and reads $\bar\zeta_{\text{R}} \in \rho (P-Q)^\#M H^{-1} c+{\cal N}$.
Note that multiplication by $(I-2P)$ leaves the space $\cal N$ invariant by Lemma~\ref{lm-N}.
By multiplying both sides of the above equality by $(I-2P)$, we obtain $\bar\zeta\in \rho (I-2P)(P-Q)^\#M H^{-1} c+{\cal N}$
which proves Lemma~\ref{cvg-quad}.}
\end{proof}

We are now in position to prove Theorem~\ref{main} in the quadratic case. 
Item \emph{i)} in the statement of this theorem was shown by
Lemma \ref{vap-1}.  

Given any fixed point $\zeta_\star$ of the transformation $\zeta' = R \zeta 
+ d$, we have 
\begin{equation}
\label{iterees-zeta}
\zeta_k - \zeta_\star = R (\zeta_{k-1} - \zeta_\star) = \cdots = 
R^k (\zeta_{0} - \zeta_\star) .
\end{equation} 
Defining $\xlim = - H^{-1}M^*(I-2P)\zeta_\star - H^{-1}c$ and recalling 
Eq.~\eqref{x-quad}, we have
\begin{align*}
x_{k+1} - \xlim &= - H^{-1} M^* (I - 2P)( \zeta_k - \zeta_\star ) \\
                  &= - H^{-1} M^* (I - 2P)R^k (\zeta_{0} - \zeta_\star) \\
&= - H^{-1} M^* (I - 2P)
\left( \Pi_{\cal N} + W_2 \tilde\Lambda^k \underline W_2^* \right) 
(\zeta_{0} - \zeta_\star)\,.
\end{align*}
By Lemma~\ref{lm-N}, note that $P\Pi_{\cal N} = Q\Pi_{\cal N}=0$. 
\modifwh{Since $\colsp(Q) = \colsp(M)$, we have $M^* (I-2P) \Pi_{\cal N}= 0$.}
Consequently, 
\begin{equation}
\label{eq:vit_quad}
x_{k+1} - \xlim = - H^{-1} M^* (I - 2P) W_2
\tilde\Lambda^k \underline W_2^* (\zeta_{0} - \zeta_\star)\,.
\end{equation}
Therefore, $(x_{k})$ converges to $\xlim$ as $k\to\infty$.
Since we know already by Theorem~\ref{1st-ord} that $(x_{k})$
converges to $\un_N\otimes x_\star$, this implies that
$\xlim=\un_N \otimes x_\star$.  It is worth noting that this
identity could have been be derived directly from the mere definition of
$\xlim$ with no need to use Theorem~\ref{1st-ord}. As a sanity
check, the reader may indeed verify that $\xlim=\un_N \otimes
x_\star$ using direct algebra. We skip this verification here as it is
not mandatory for the proof.  

Equality~(\ref{eq:vit_quad}) yields Theorem~\ref{main}-\emph{ii)} in 
the quadratic case. 

To show Theorem~\ref{main}-\emph{iii)}, write 
\[
W_2 \tilde\Lambda^k \underline W_2^* = 
\Bigl[ W_{2,1} \ W_{2,2} \Bigr] 
\begin{bmatrix} \tilde\Lambda_1^k & \\ & \tilde\Lambda_2^k \end{bmatrix} 
\begin{bmatrix} \ \underline W_{2,1}^* \ \\ \ \underline W_{2,2}^* \ 
\end{bmatrix} 
\]
where $\tilde\Lambda_1$ collects on its diagonal the eigenvalues of 
$W_2 \tilde\Lambda \underline W_2^*$ with the absolute value $\alpha$, and
write 
\begin{align*} 
x_{k+1} - \xlim &= 
-H^{-1} M^* (I-2P) W_{2,1} \tilde\Lambda_1^k \underline W_{2,1}^* 
       (\zeta_0 - \zeta_\star) + \xi_k \\
&= G \tilde\Lambda_1^k v + \xi_{k} 
\end{align*} 
where $G=-H^{-1} M^* (I-2P) W_{2,1}$, $v=\underline W_{2,1}^*(\zeta_0 - \zeta_\star)$
and $\xi_k=-H^{-1} M^* (I-2P) W_{2,2}\Lambda_2^k
\underline W_{2,2}^*(\zeta_0 - \zeta_\star)$\,.
The condition on $\colsp((I-2P)M)$ in the statement of Theorem \ref{main} 
asserts that $G \neq 0$. 
On the other hand, 
$v = \underline W_{2,1}^* (\zeta_0 - \zetalim)$ by Lemmas \ref{vap-1}
and \ref{cvg-quad}. Hence, when $\zeta_0$ lies outside a set with zero 
Lebesgue measure as we shall assume, $Gv \neq 0$.
Denote by $\alpha e^{\imath\theta_1},\dots,\alpha e^{\imath\theta_L}$ the distinct 
elements on the diagonal of $\tilde\Lambda_1$.
For any $k\geq 0$, the vector $G\tilde\Lambda_1^kv$ coincides with
$\alpha^k f(k)$ where $f(k)=\sum_{\ell=1}^L a_\ell e^{\imath\theta_\ell k}$
for some coefficients $a_1,\dots,a_L$. Since $Gv\neq 0$, at least one of these 
coefficients is non zero. Hence, 
$\lim\sup_k | f(k) | > 0$.
Indeed, one can easily show that $n^{-1} \sum_{k=0}^{n-1} | f(k) |^2 
\xrightarrow[n\to\infty]{} \sum | a_\ell |^2  > 0$. This would not be possible
if $\lim\sup_k | f(k) | = 0$.  

By construction, $\sprad(\tilde\Lambda_2) < \alpha$, hence 
$\| \xi_k \| \leq C \beta^k$ where $C$ is a constant and where 
$0 \leq \beta < \alpha$. We therefore have 
\[
\| x_{k+1} - \xlim \| \geq \| G \tilde\Lambda_1^k v \| - \| \xi_k \| 
\geq \alpha^k ( |f(k)| - C (\beta / \alpha)^k ) 
\]
and we obtain that $\| x_{k+1} - \xlim \| \geq \alpha^k g(k)$ where 
$g(k) = \max(|f(k)| - C (\beta / \alpha)^k, 0)$.
Observing that $\limsup_k g(k) > 0$ and using the convention 
$\log 0 = -\infty$, we obtain 
\begin{align*} 
\lim\sup_k \frac 1k \log \| x_{k+1} - \xlim \| &\geq 
\log\alpha + \lim\sup_k \frac 1k \log g(k) \\
&= \log\alpha .
\end{align*} 
Combining this lower bound with the already established upper bound, we obtain 
the result.

\subsection{The general case} 
\label{subsec-general}

We now assume that the functions $f_n$ satisfy Assumptions \ref{conv-lsc} and
\ref{nabla2} in full generality. Theorem~\ref{1st-ord} shows that the 
iterates $x_k$ converge towards the unique minimizer 
$\xlim = \un_N \otimes x_\star$ of the problem. 
For $k$ large enough, the iterates $x_k$ are in a neighborhood of $\xlim$
where the functions $f_n$ are differentiable, and the $x$ -- update 
equation (Eq.~\eqref{xk}) boils down to the equation 
$\nabla f(x_{k+1}) + \rho M^* M x_{k+1} = - M^* (\lambda_k -
\rho z_k) = -M^* (I-2P)\zeta_k$. This equation can be rewritten in two
different manners. On the one hand, we have 
\[
x_{k+1} = \text{prox}_h( - (\rho M^* M)^{-1} M^* (I-2P) \zeta_k) 
\] 
where $h(x) = \sum_{n=1}^N [(\rho M^* M)^{-1}]_{nn} f_n(x(n))$, and 
on the other hand, we have for any $x$ close enough to $\xlim$ 
\[
\nabla f(x) 
= \nabla f(\xlim) + \nabla^2 f(\xlim)(x-\xlim) + E(x-\xlim)  
\] 
where $\| E(x) \| / \| x \|\to 0$ as $x \to 0$.  
With this relation, the update equation for $x$ becomes 
\begin{align*} 
(\nabla^2 f(\xlim) + \rho M^* M) x_{k+1} &= -M^* (I-2P) \zeta_k - c \\ 
&\phantom{=} - E(x_{k+1} - \xlim) 
\end{align*} 
where $c = \nabla f(\xlim) - \nabla^2 f(\xlim) \xlim$, or equivalently
\begin{align} 
x_{k+1} &= - H^{-1} M^* (I-2P) \zeta_k - H^{-1} c \nonumber \\ 
&\phantom{=} - H^{-1} E(x_{k+1}- \xlim) 
\label{iterees-x} 
\end{align} 
where $H = \nabla^2 f(\xlim) + \rho M^* M$. Mimicking the derivation made 
before Remark~\ref{DR}, this equation leads to 
\begin{equation}
\label{zeta-general} 
\zeta_{k+1} = R\zeta_k + d - \rho M H^{-1} E(x_{k+1} - \xlim) 
\end{equation} 
where $R = (P_\perp - Q)(I- 2P)$ with $Q = \rho M H^{-1} M^*$ as in 
Equation \eqref{def-Q}, and where $d = - \rho M H^{-1} c$. \\ 
By replacing the matrix 
$\Phi$ defined in \eqref{def-Sigma} with $\nabla^2 f(\xlim)$, we notice 
that the lemmas \ref{lm-R}--\ref{cvg-quad} remain true for the matrices 
$R$ and $Q$ just introduced. 
Moreover, Equation \eqref{zeta-general} shows that the sequence $\zeta_k$
converges to a fixed point of the transformation $\zeta' = R \zeta + d$. 
Making $k\to\infty$ in \eqref{iterees-x}, we also notice that 
\begin{align}
\xlim &= - H^{-1} M^* (I-2P) \zeta_\star - H^{-1} c \label{xlim} \\
&= 
\text{prox}_{h}( - (\rho M^* M)^{-1} M^* (I-2P) \zeta_\star)
\nonumber 
\end{align} 
where $\zeta_\star$ is any fixed point of the transformation 
$\zeta' = R \zeta + d$. \\
We now have the elements to establish the Theorem~\ref{main}-\emph{ii)}. 
Given one fixed point $\zeta_\star$, the analogue of Eq. \eqref{iterees-zeta} 
is 
\begin{equation}
\label{delta-zeta} 
\zeta_k - \zeta_\star = R^k (\zeta_0 - \zeta_\star)  
- \rho \sum_{\ell=1}^k R^{k-\ell} M H^{-1} E(x_\ell - \xlim) . 
\end{equation} 
Lemma~\ref{lipsh} shows now that 
\begin{align*} 
& \| x_{k+1} - \xlim \| \\
=& \| \text{prox}_h( - (\rho M^* M)^{-1} M^* (I-2P) \zeta_k ) \\ 
& \ \ \ \ 
- \text{prox}_h( - (\rho M^* M)^{-1} M^* (I-2P) \zeta_\star ) \| \\
\leq& \| (\rho M^* M)^{-1} M^* (I-2P) (\zeta_k - \zeta_\star) \| \\
\leq& \| (\rho M^* M)^{-1} \| \ \Bigl( 
   \| M^* (I-2P) R^k (\zeta_0 - \zeta_\star) \| \\
& + \rho \sum_{\ell=1}^k \| M^* (I-2P) R^{k-\ell} 
M H^{-1} E(x_\ell - \xlim) \| \Bigr) .  
\end{align*} 
Our purpose is to show that 
\begin{equation}
\label{borne-sup}
\forall \, \varepsilon > 0, \ 
\sup_k (\alpha+\varepsilon)^{-k} \| x_{k+1} - \xlim \| < \infty .
\end{equation} 
This shows indeed that 
\[
\limsup_k \frac{\log\| x_{k+1} - \xlim \|}{k} \leq \log\alpha 
+ \log(1 + \frac{\varepsilon}{\alpha}) 
\]
for any $\varepsilon > 0$, which is equivalent to 
Theorem~\ref{main}-\emph{ii)}. 

Fix $\varepsilon > 0$. Recall that $\| M^* (I-2P) R^k \| \leq C \alpha^k$
where $C$ is a constant, $x_k \to \xlim$, and $\| E(x_k - \xlim) \| = 
o(\| x_k - \xlim \|)$. 
By delaying the time origin as much as needed, we can assume that 
for any $k > 0$ and any $\ell \in \{1,\ldots, k\}$, 
\begin{gather*} 
\| (\rho M^* M)^{-1} \| \, \| M^* (I-2P) R^k (\zeta_0 - \zeta_\star) \| 
\leq \alpha^{k+1}, \\
\rho \| (\rho M^* M)^{-1} \| \, \| M^* (I-2P) R^{k-\ell}
M H^{-1} \| \leq \alpha^{k+1-\ell},  \\
\| E(x_\ell- \xlim) \| \leq \delta \| x_\ell - \xlim \|,  \ \text{and} \\
\| x_0 - \xlim \| \leq B = 
\frac{\varepsilon}{\varepsilon - \delta(\alpha+\varepsilon)} 
\end{gather*} 
where we choose $\delta < \varepsilon / (\alpha+\varepsilon)$. With this 
choice of the time origin, we have 
\[
\| x_{k+1} - \xlim \| \leq \alpha^{k+1} + \delta \sum_{\ell=1}^k 
\alpha^{k+1-\ell} \| x_{\ell} - \xlim \| . 
\]
Putting $w_k = (\alpha+\varepsilon)^{-k} \| x_{k} - \xlim \|$, this inequality
is rewritten 
\[
w_{k+1} \leq \Bigl(\frac{\alpha}{\alpha+\varepsilon}\Bigr)^{k+1} 
+ \delta \sum_{\ell=0}^k  
\Bigl(\frac{\alpha}{\alpha+\varepsilon}\Bigr)^{k+1-\ell} w_\ell .
\]
We know that $w_0 \leq B$. Assume that $w_1, \ldots, w_k \leq B$. Then
\[
w_{k+1} < 1 + \delta B \sum_{\ell=0}^k 
\Bigl(\frac{\alpha}{\alpha+\varepsilon}\Bigr)^{k+1-\ell} 
< 1 + \frac{\delta B}{1 - \frac{\alpha}{\alpha+\varepsilon}} 
= B
\]
and Inequality \eqref{borne-sup} is established. 

We now show Theorem~\ref{main}-\emph{iii)}. 
From the equations \eqref{iterees-x}, \eqref{xlim} and \eqref{delta-zeta}, 
we have 
\begin{align*}
x_{k+1} - \xlim &= 
-H^{-1} M^* (I-2P) R^k (\zeta_0 - \zeta_\star) \\
&\phantom{=} 
+ \rho \sum_{\ell=1}^k 
H^{-1} M^* (I-2P) R^{k-\ell} M H^{-1} E(x_\ell - \xlim) \\
&\phantom{=} - H^{-1} E(x_{k+1} - \xlim) \\
&= X_k + Y_k + Z_k . 
\end{align*}
By the argument establishing Theorem~\ref{main}-\emph{iii)} in the quadratic 
case, for any $\zeta_0$ outside a set of Lebesgue measure zero, there is a 
function $g(k)$ such that $\| X_k \| \geq \alpha^k g(k)$ and 
$a = \limsup_k g(k) > 0$. For the sake of contradiction, assume that 
$\limsup_k ( k^{-1} \log \| x_k - \xlim \| ) < \log\alpha$ for this $\zeta_0$. 
Then $\| x_k - \xlim \| \leq C \beta^k$ for some $C > 0$ and some 
$\beta \in (0, \alpha)$. By delaying the time origin as much as needed, we can 
assume that 
\begin{align*} 
\frac{\| E(x_\ell - \xlim) \|}{\| x_\ell - \xlim \|} &\leq 
\delta = \frac{a(\alpha - \beta)}{2\alpha} \ \text{for any } \ell > 0, \ 
\text{and} \\ 
\| Y_k + Z_k \| &\leq 
\delta \sum_{\ell=1}^k \alpha^{k-\ell} \beta^\ell   
           \ + \ \delta \beta^k    \\ 
&< \delta \frac{\alpha}{\alpha - \beta} \alpha^k . 
\end{align*} 
We therefore have 
\[
\| x_{k+1} - \xlim \| \geq \| X_k \| - \| Y_k + Z_k \| 
\geq \alpha^k \Bigl( g(k) - \delta \frac{\alpha}{\alpha - \beta} \Bigr) . 
\]
Since $\limsup_k ( g(k) - \delta \alpha / (\alpha - \beta) ) = a/2 > 0$, 
we obtain $\limsup_k ( k^{-1} \log \| x_k - \xlim \| ) \geq \log\alpha$. 
Theorem~\ref{main}-\emph{iii)} is proven.

\section{Numerical Illustrations}
\label{sec-simus}

We first provide a numerical illustration in the special cases
described in Section~\ref{subsec-centralized}
and~\ref{subsec-ring}. 
The second order derivative $\sigma_\star^2$ of
the functions $f_n$ at the minimum is set to
$16$. Figures~\ref{fig:alpha_centralise} and~\ref{fig:alpha_ring}
represent the rate $\alpha$ as a function of the step-size $\rho$ of
the algorithm in the case of a centralized network and a ring network
respectively.  
\begin{figure}[h]
  \centering
  \includegraphics[width=.9\linewidth]{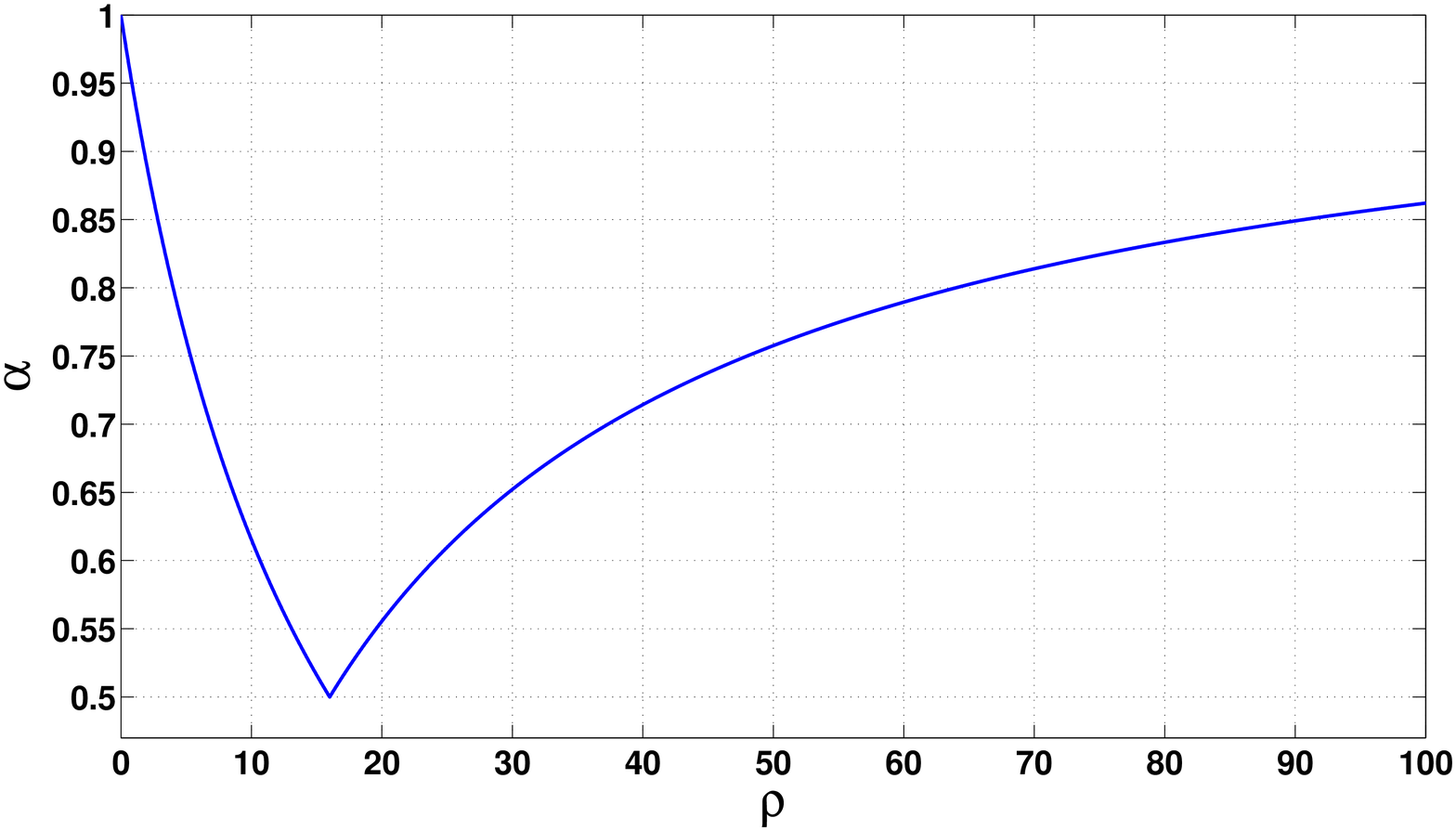}
  \caption{Rate $\alpha$ as a function of $\rho$ - Centralized network - $\sigma_\star^2=16$.}
  \label{fig:alpha_centralise}
\end{figure}
\begin{figure}[h]
  \centering
  \includegraphics[width=.9\linewidth]{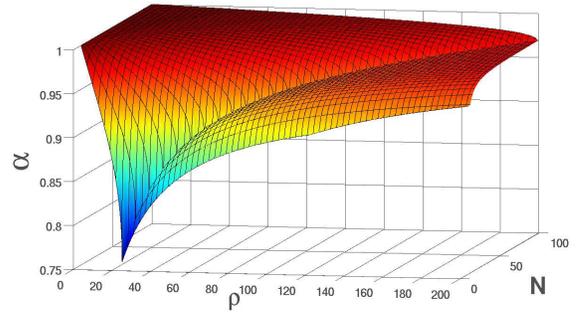}
  \caption{Rate $\alpha$ as a function of $\rho$ and $N$ - Ring network - $\sigma_\star^2=16$.}
  \label{fig:alpha_ring}
\end{figure}
In the centralized case, the optimal value of $\rho$
coincides with $\sigma^2$ and is thus independent of $N$.  In the ring
network, the rate $\alpha$ depends on both $\rho$ and $N$.

We now address the case where the second order derivatives are not necessarily equal. 
We set $N=5$ and assume that the values of $f''_n(x_\star)$ for all agents $n$ are respectively equal to
4, 9, 16, 25 and 39. 
\begin{figure}[h]
  \centering
  \includegraphics[width=.9\linewidth]{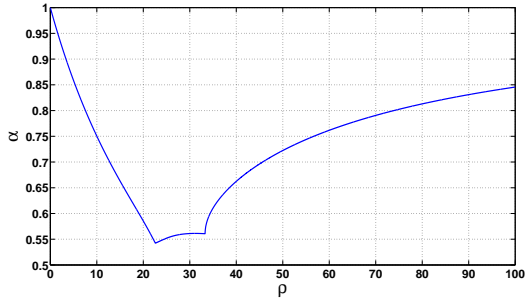}
  \caption{Rate $\alpha$ as a function of $\rho$ - Centralized network - $N=5$ - Distinct second order derivatives.}
  \label{fig:alpha_centralise_nonDiagonal}
\end{figure}
Figure~\ref{fig:alpha_centralise_nonDiagonal} represents the rate $\alpha$ as a function of $\rho$.

Finally, we compare our theoretical result with the performance of ADMM observed by simulation; we also compare with the bound of Shi et al.~\cite{shi-etal-(arxiv)13}.
These simulations were conducted on a $20$ nodes Random Geometric Graph with radius $0.2$.
The sets $\{A_\ell\}_{\ell}$ are taken equal to the pairs of connected agents as in Example 2 of Section~\ref{sec:examplesCentralizedRing}.
We plot  $-k^{-1}\log \|x_k-x_\star\|$ as a function of the number of iterations $k$.

In Figure~\ref{fig:simu_exp}, the functions are taken as $f_n(x) = \exp(\beta_n x)$ where the $\beta_n$'s are drawn uniformly in $[-10,10]$ then centered (in that case $\sum_n f_n$ admits $x_\star=0$ as unique minimizer) and $\rho$ is set to $20$.
Here, the bound of  Shi et al.~\cite{shi-etal-(arxiv)13} is not defined (\emph{i.e.}, its $\log$ is equal to zero).
As expected, Figure~\ref{fig:simu_exp} shows that the rate $\alpha$ is tight in the sense that $-k^{-1}\log \|x_k-x_\star\|$
numerically converges to $-\log \alpha$. \begin{figure}[h]
  \centering
  \includegraphics[width=.9\linewidth]{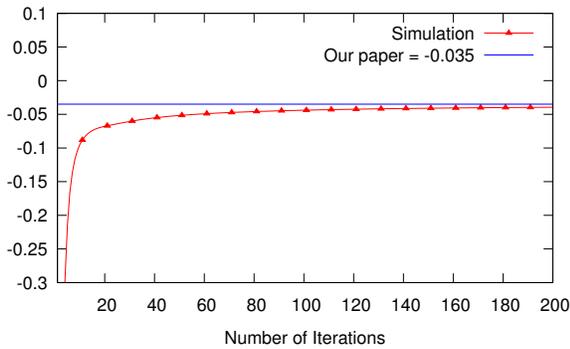}
  \caption{ $k^{-1}\log \|x_k-x_\star\|$ as a function of the number of iterations $k$ - $N=20$ - Exponential functions}
  \label{fig:simu_exp}
\end{figure}

We also investigate the case of quadratic functions. In that case, the bound of Shi et al.~\cite{shi-etal-(arxiv)13} is well defined
and plotted in Figure~\ref{fig:simu_quad}. The functions $f_n$ are defined as $f_n(x) =  a_n(x-b_n)^2$ where the $a_n$'s are drawn uniformly in $[1,100]$ and the $b_n$'s are drawn from a Gaussian distribution with mean $5$ and variance $100$. The parameter $\rho$ has been set to $100$ as this seems to be a good choice to take it around the second order derivatives from the above simulations and derivations. 
We observe that our characterization of the convergence rate is tight in the sense that it fits the empirical performance of ADMM, 
whereas a gap exists between the latter and the bound of \cite{shi-etal-(arxiv)13}. 
\begin{figure}[h]
  \centering
  \includegraphics[width=.9\linewidth]{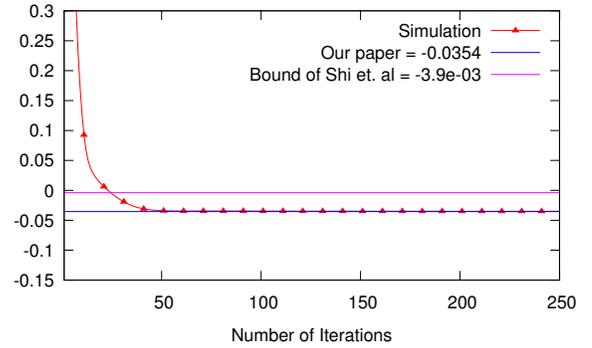}
  \caption{ $k^{-1}\log \|x_k-x_\star\|$ as a function of the number of iterations $k$ - $N=20$ - Quadratic functions}
  \label{fig:simu_quad}
\end{figure}

\section{Conclusion}
\label{sec:conclusion}

In this paper, we addressed the rate of convergence of ADMM to solve
distributively the minimization problem $\inf_{x} \sum_{n = 1}^N f_n(x)$ where 
the $f_n$'s are private convex functions. Letting $\xlim$ be the minimizer and
assuming that the functions are twice differentiable at $\xlim$, we obtained an
explicit characterization of the linear convergence rate under the form
of the spectral radius of a matrix depending to the Hessian
$\sum_n \nabla^2f_n(\xlim)$ and the communication network.
Under mild conditions, it is shown that the obtained rate is tight in the sense
that the actual convergence rate is no faster than the one obtained.

In practice, our analysis is useful to accurately predict the performance of ADMM and
to optimize various parameters, such as the step-size of the algorithm and the 
communication graph. Our method potentially allows to design augmented ADM methods
with enhanced convergence rate.

\bibliographystyle{IEEEbib}
\bibliography{math}

\appendix

\section{Proof of Corollary~\ref{coro:ring}}
\label{sec:ring}

We start by introducing a notation and by recalling a known fact. 
Given three $2 \times 2$ matrices $G_{-1}$, $G_0$ and $G_1$, we denote as 
$A = \circul_N(G_0 + e^{\imath\lambda} G_1 + e^{-\imath\lambda} G_{-1})$  
the $2N \times 2N$ block circulant matrix 
\[
A = \begin{bmatrix} 
G_0    & G_1    &       &   G_{-1} \\ 
G_{-1} &  G_0   & G_1 \\ 
       &        &   \ddots \\ 
G_1    &        & G_{-1} & G_0 
\end{bmatrix} . 
\]
By generalizing a well known result for circulant matrices 
(see \emph{e.g.}~\cite{gray}), we know that the eigenvalue spectrum of $A$ 
coincides with the set of eigenvalues of the trigonometric matrix polynomial 
$ {\cal G}(e^{\imath\lambda}) = 
G_0 + \exp(-\imath\lambda) G_{-1}  + \exp(\imath\lambda)  G_{1} 
$
taken at $\lambda = 2\pi k / N$ for $k = 0,\ldots, N-1$. \\
Getting back to our model, we have $T = 2N$, 
\[
M = \begin{bmatrix} 
1 \\
& 1 \\
& 1 \\
&  & 1 \\
&  & 1 \\
&  &  & \ddots \\
&  &  &       & 1 \\
&  &  &       &  & 1 \\
&  &  &       &  & 1 \\
1 
\end{bmatrix} 
\]
and 
\[
P = \frac 12 \begin{bmatrix} 
1 & 1 \\
1 & 1 \\
  &  & 1 & 1 \\ 
  &  & 1 & 1 \\ 
  &  &   &  & \ddots \\
  &  &   &  &  \\
  &  &   &  &  & 1 & 1 \\
  &  &   &  &  & 1 & 1 \\
  &  &   &  &  &   &  &1 & 1 \\
  &  &   &  &  &   &  &1 & 1 
\end{bmatrix} . 
\]
Since $M^* M = 2 I_N$, we also have 
\begin{align*} 
Q &= \frac{\rho}{\sigma_\star^2 + 2 \rho} \begin{bmatrix}
1&   &  &   &  &    & 1 \\ 
 & 1 & 1 \\
 & 1 & 1 \\
 &   &  & \ddots \\ 
 &   &  &       & 1 & 1 \\
 &   &  &       & 1 & 1 \\
1&   &  &       &   &   & 1 
\end{bmatrix} \\
&= \frac{\rho}{\sigma_\star^2 + 2 \rho} 
\text{circ}_N\Bigl( I_2 + 
e^{\imath\lambda} \begin{bmatrix} 0& 0 \\ 1 & 0\end{bmatrix} 
+ e^{-\imath\lambda} \begin{bmatrix} 0& 1 \\ 0 & 0\end{bmatrix} \Bigr) . 
\end{align*} 
Noticing that $R = I - P - Q + 2QP$ and writing 
$a={\rho}/({\sigma_\star^2 + 2 \rho})$, we have 
\begin{align*}
2QP &= a \ 
\text{circ}_N\Bigl( \Bigl( I_2 + 
e^{\imath\lambda} \begin{bmatrix} 0& 0 \\ 1 & 0\end{bmatrix} + 
e^{-\imath\lambda}\begin{bmatrix} 0& 1 \\ 0 & 0\end{bmatrix} \Bigr) 
\un\un^*\Bigr) \\
&= a \ 
\text{circ}_N\Bigl( \begin{bmatrix} 1 & 1 \\ 1 & 1 \end{bmatrix} + 
e^{\imath\lambda} \begin{bmatrix} 0& 0 \\ 1 & 1\end{bmatrix} +  
e^{-\imath\lambda}\begin{bmatrix} 1& 1 \\ 0 & 0\end{bmatrix} \Bigr). 
\end{align*}
Therefore, $R = \circul_N( G_0 + e^{\imath\lambda} G_1 + 
e^{-\imath\lambda} G_{-1})$ where 
\begin{align*}
G_0 &= I_2 - \frac 12 \un\un^* - a I_2 + a \un\un^*  
= \frac 12 \begin{bmatrix} 1 & -\frac{\sigma_\star^2}{\sigma_\star^2+2\rho} \\ 
-\frac{\sigma_\star^2}{\sigma_\star^2+2\rho}  & 1 \end{bmatrix},  \\
G_1 &= \frac{\rho}{\sigma_\star^2 + 2 \rho} \Bigl( 
\begin{bmatrix} 0 & 0 \\ -1 & 0 \end{bmatrix} + 
\begin{bmatrix} 0 & 0 \\ 1 & 1 \end{bmatrix} \Bigr) 
= \frac{\rho}{\sigma_\star^2 + 2 \rho} 
\begin{bmatrix} 0 & 0 \\ 0 & 1 \end{bmatrix},  \\ 
G_{-1} &= \frac{\rho}{\sigma_\star^2 + 2 \rho} \Bigl( 
\begin{bmatrix} 0 & -1 \\ 0 & 0 \end{bmatrix} + 
\begin{bmatrix} 1 & 1 \\ 0 & 0 \end{bmatrix} \Bigr) 
= \frac{\rho}{\sigma_\star^2 + 2 \rho} 
\begin{bmatrix} 1 & 0 \\ 0 & 0 \end{bmatrix} . 
\end{align*} 
The eigenvalues of ${\cal G}(e^{2\imath\pi k/N}) = G_0 
+ e^{2\imath\pi k/N} G_1 + e^{-2\imath\pi k/N} G_{-1}$ are the solutions
of the equation $\lambda^2 - \lambda s_k + D_k = 0$ where 
\[
s_k = \tr {\cal G}(e^{2\imath\pi k/N}) = 
\frac{\sigma_{\star}^2 + 2\rho \left(1 + \cos(2\pi k/N) \right)}
{\sigma_{\star}^2 + 2\rho} 
\]
and 
\[
D_k = \det {\cal G}(e^{2\imath\pi k/N}) = 
\frac{\rho}{\sigma_\star^2 + 2\rho} 
\left( 1 + \cos(2\pi k/N) \right) . 
\]
The analysis of these solutions for all $k=0,\ldots, N-1$, which is tedious but straightforward, 
directly leads to the expression of $\alpha$ in Corollary~\ref{coro:ring}.
By simple algebra, the rate $\alpha=\alpha(\rho)$ is shown to be a continuous function of $\rho$
which is decreasing on the interval $(0,\frac{\sigma_\star^2}{2s_N}]$ and increasing on $[\frac{\sigma_\star^2}{2s_N},+\infty)$.
\end{document}